\documentclass[12pt, draftclsnofoot, onecolumn]{IEEEtran}

\usepackage[numbers,sort&compress,square]{natbib}   
\usepackage{graphicx}
\usepackage{subfigure}
\usepackage{amsmath,amssymb}
\usepackage{amsfonts}
\usepackage{array}
\usepackage{amsthm}
\usepackage{mathrsfs}
\usepackage{amsmath}
\usepackage{setspace}
\usepackage{color}
\usepackage[ruled,linesnumbered]{algorithm2e}
\usepackage{epstopdf}
\usepackage{soul}

\newcommand{\beq}{\begin{equation}}
 \newcommand{\eeq}{\end{equation}}
  
 \newcommand{\bea}{\begin{eqnarray}}
 \newcommand{\eea}{\end{eqnarray}}

\renewcommand{\IEEEQED}{\IEEEQEDclosed}
\newtheorem{theorem}{Theorem}

\newtheorem{lemma}[theorem]{Lemma}

\theoremstyle{definition}
\newtheorem{define}[theorem]{Definition}

\newtheorem{remark}[theorem]{Remark}
\newtheorem{example}[theorem]{Example}

\usepackage{color}

\renewcommand{\baselinestretch}{1.4}

\begin{document}

\title{\huge To Motivate Social Grouping in Wireless Networks}

\author{\normalsize \IEEEauthorblockN{Yu-Pin Hsu\IEEEauthorrefmark{1} and
Lingjie Duan\IEEEauthorrefmark{2}\\}
\IEEEauthorblockA{\IEEEauthorrefmark{1}Laboratory for Information and Decision Systems, Massachusetts Institute of Technology\\}
\IEEEauthorblockA{\IEEEauthorrefmark{2}Engineering Systems and Design Pillar, Singapore University of Technology and Design\\}
\IEEEauthorblockA{{yupinhsu@mit.edu, lingjie\_duan@sutd.edu.sg}}
}


\maketitle

%
%

\begin{abstract}
We consider a group of neighboring smartphone users who are roughly at the same time interested in the same network content, called \textit{common interests}.
However, ever-increasing data traffic challenges the limited capacity of base-stations (BSs) in wireless networks.   To better utilize the limited BSs' resources under  \textit{unreliable} wireless networks,  we propose local common-interests sharing (enabled by D2D communications) by \textit{motivating} the physically neighboring users  to form a \textit{social group}. 
As users are selfish in practice, an incentive mechanism is needed to motivate social grouping. We  propose a novel concept of \textit{equal-reciprocal} incentive over \textit{broadcast} communications, which fairly ensures that each pair of the users in the social group share  the same amount of content with each other. 
As the equal-reciprocal incentive may restrict the amount of content shared among the users,  we analyze the \textit{optimal} equal-reciprocal scheme that maximizes local sharing content. While ensuring fairness among users, we show that this optimized scheme also maximizes each user's utility in the social group.   Finally, we look at \textit{dynamic} content arrivals  and extend our scheme successfully by proposing novel \textit{on-line} scheduling algorithms. 
\end{abstract}


\section{Introduction}\label{section:intro}

Billions of people around the globe rely on wireless devices for  conferencing,  streaming, and file downloading. Unfortunately, wireless networks are inherently \textit{less reliable} than wired networks. Moreover, the \textit{limited} usable spectrum of base-stations (BSs) poses significant challenges to network designers.  In this paper, we are addressing the key issues of unreliability and limited resource in wireless networks by  leveraging the \textit{broadcast} nature of wireless communications. 


%
%

To that end,  we first notice that  video traffic will increase to 72\% of the total traffic by 2020 \cite{cisco-2015}. Moreover, it becomes more and more popular that people watch videos on their own smartphones or personal devices \textit{individually}. A group of friends would be watching \textit{live} sport-games (like football or soccer) or TV programs together, such as at a bar, a camping site, or a bus. In fact, half of males, aged 18-34, look at videos  with friends in person according to \cite{google}. 

In particular, we are motivated by the \textit{MicroCast} system, recently proposed and implemented as an Android application  (e.g., see \cite{hulya-1,hulya-2,hulya-3}). In the MicroCast scenario,  a small group of smartphone users within the proximity of each other are almost at the same time interested in the same content. We refer to the same content for different users as  \textit{common interests}. We remark that due to growing attention of social networks \cite{social:Li}, which exhibit homophily by sharing common interests and similar behaviors, a Google patent \cite{d2d-patent} defines a device-to-device (D2D) service type that finds people who share common interests. 
 
Traditionally  users communicate with   BSs via their cellular links \textit{independently}, i.e., the BSs unicast  streaming to the users. Ideally, if the channels are noiseless, identifying the common interests of the users can save the time to deliver the common interests by leveraging the \textit{broadcast} BS-to-device (B2D) medium. However, because of unreliable wireless channels, we first show that the broadcast B2D (i.e., identifying the common interests) is not that effective unless the local users are willing to share the common interests together.   

Due to the emerging \textit{broadcast} D2D technology,  neighboring users that are physically close  can communicate with each other. Through local content sharing over D2D communications,  the BSs' bandwidth can be saved. For example, if a user in Fig. \ref{fig:motivation} loses content from the BS's broadcast, the user can still obtain it from other users via D2D sharing.  We consider \textit{hybrid networks} including both  the \textit{broadcast} B2D  and \textit{broadcast} D2D communication technologies, as shown in Fig. \ref{fig:motivation}.  In other words, we are leveraging both femto-caching \cite{femto-caching} and D2D-caching \cite{d2d-caching}, where the femto-caches in BSs store common interests for users, while some information can be stored in D2D-chases of local users for possible future sharing.

\begin{figure}[t]
\centering
\includegraphics[width=.5\textwidth]{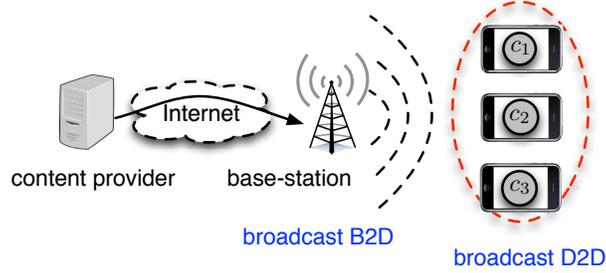}
\caption{Co-existence of broadcast BS-to-device (B2D) and   broadcast device-to-device (D2D) communications. Users $c_1, c_2, c_3$ are simultaneously interested in   the same streaming.}
\label{fig:motivation}
\end{figure}

%
%

However, the users are selfish in general; as such, we cannot expect that the selfish users are willing to share their own resources. To facilitate D2D communications by building up a social relationship, there is a critical need to design an incentive mechanism that provides the neighboring users with a motivation to form a social group. Moreover, the users  would appreciate a free and simple sharing service without complex money exchanges.

Hence, we propose a novel \textit{non-monentary} incentive concept, which is referred to as \textit{equal-reciprocal} scheme.
Applying the equal-reciprocal incentive over the \textit{broadcast} hybrid networks,  the BS impartially requests each pair of the users in the social group to share the same  amount of content with each other. 



Considering users' heterogeneous B2D channels, as the equal-reciprocal incentive may limit the amount of content the users are willing to share,   more questions arise in order to optimize  network performance (e.g., throughput or delays): (1) \textit{Which users should share their resource?} (2) \textit{How much content should the users share with others}?  We will answer the questions by analyzing the \textit{optimal} equal-reciprocal scheme that optimizes network performance.

Furthermore, we  show that  the BS and the local users have a \textit{win-win} situation by using the proposed optimal equal-reciprocal scheme. The optimal equal-reciprocal scheme minimizes the expected  time to  deliver the  common interests to all users, and at the same time it maximizes a utility of each user in the group, where the utility is the difference between the amount of content uploaded from this user and those downloaded from the social group. In other words,  we are simultaneously solving both \textit{global} and \textit{local} optimization problems. Hence, in addition to the fairness resulted from the equal-reciprocal idea, the optimal equal-reciprocal scheme also maximizes each user's utility.

Finally, we propose \textit{on-line scheduling algorithms}, including both optimal centralized and sub-optimal distributed algorithms (easy to implement) to address a more general case of dynamic content arrivals. The algorithms adaptively select a user to share  received content in a timely fashion. We will analyze the proposed algorithms from both theoretical and numerical perspectives. 

\subsection{Related work} \label{section:related-work}
In this paper, we focus on motivating social grouping over the  hybrid networks. We explore several related research directions as follows.  

\subsubsection{Local repair on multicast} 
      As our network scenarios have the notion of cooperatively recovering lost content, we first compare our work with local repair schemes for the multicast routing  (e.g., MAODV \cite{maodv}). 
Lots of such local repair schemes are proposed to find a substitute for a failure link; however, there is no notion of content sharing among a \textit{group} of users. In this paper, we consider a group of neighboring users with a potential to share content, where an incentive is required.

\subsubsection{D2D networks}
Recently, there is an increasing concern about  D2D communications, e.g., see the survey papers \cite{survey:liu,survey:Al-Kanj} and the references therein. Many work about D2D networks ignores  broadcast advantages (as also stated in \cite{abedini:streamming}), while  the most related work, which has the notions of multicast, broadcast, or social group, is \cite{hulya-1, hulya-2, hulya-3, abedini:streamming,cooperation:Seppala,cooperation:Andreev,
cooperation:lin,cooperation:Chen,cooperation-coding:Liu-1,cooperation-coding:Liu-2,cooperation-coding:alex-1,cooperation-coding:alex-2,cooperation-coding:yupin,cooperation-incentive:Li,cooperation-incentive:Cao}. 

A reliable D2D multicast notion is introduced in \cite{cooperation:Seppala}. The authors in \cite{cooperation:Andreev} analyze the performance of D2D-assisted networks using stochastic geometry. In \cite{cooperation:lin}, a D2D network is studied in terms of   optimizing multicast by choosing the optimal multicast rate and the optimal number of re-transmission times, while  \cite{cooperation:Chen} proposes a social group utility maximization framework. Moreover, there is some network coding design for D2D-assisted networks, e.g., \cite{cooperation-coding:Liu-1,cooperation-coding:Liu-2,cooperation-coding:alex-1,cooperation-coding:alex-2,cooperation-coding:yupin}. The work \cite{hulya-1, hulya-2, hulya-3, abedini:streamming, cooperation:lin,cooperation:Seppala,cooperation:Andreev,cooperation:Chen,cooperation-coding:Liu-1,cooperation-coding:Liu-2,cooperation-coding:alex-1,cooperation-coding:alex-2,cooperation-coding:yupin} basically assumes neighboring users are \textit{always} willing to help each other. 

We want to emphasize that we focus on the MicroCast scenario \cite{hulya-1,hulya-2, hulya-3}, in which a group of smartphone users are interested in the same content simultaneously. Based on the scenario,  the paper  \cite{abedini:streamming} devices a real-time  scheduling algorithm. On the contrary, we are designing and analyzing an incentive mechanism for the MicroCast scenario, i.e., to motivate social grouping.

\subsubsection{Incentive design for P2P/D2D networks}

Traditional incentive  schemes for P2P networks (e.g., the  tit-for-tat incentive \cite{p2p:neely} and \cite{p2p:Kamvar,p2p:fabil,p2p:zhou}) are \textit{unicast}-based, and do not take full advantage of the broadcast feature of wireless networks. 
Of the most relevant literature about incentive design for D2D networks (which uses social group or broadcast) is  \cite{cooperation-incentive:Cao,cooperation-incentive:Li}. They propose incentive schemes based on a coalitional game and mean field game, respectively, while \cite{cooperation-incentive:Li} needs money transfers and  the incentive mechanism in \cite{cooperation-incentive:Cao} is designed for a given social relationship. 

In contrast, we  propose  a \textit{non-monetary} incentive scheme that simply leverages  \textit{broadcast} D2D. We can view our proposed equal-reciprocal incentive as a generalization of the existing tit-for-tat scheme for broadcast networks. Better exploiting the broadcast D2D,   our scheme will bring many more advantages beyond the tit-for-tat incentive; however,  more challenges arise because a broadcast packet will be received by all neighboring users in a group, so the users are co-related with
each other. Furthermore, we also consider \textit{heterogeneity} of wireless channels as well as \textit{dynamic} content arrivals.

Despite the broadcast nature of wireless communications, few incentive design in D2D networks takes advantage of the broadcast
medium. This work contributes to this emerging line of research.

\subsubsection{Network coding}
Network coding can improve bandwidth usage in unreliable networks by leveraging broadcast medium, e.g., random linear network coding \cite{net-coding}. However, we are  addressing the similar issue from another perspective. We adopt a hybrid B2D and D2D architecture, while focusing on how to promote D2D communications by incentivizing neighboring users to construct a social relationship. We note that our scheme is compatible with network coding. In other words, network performance can be further improved by using both network coding and our scheme.

%
%
%

\section{Symmetric networks} \label{section:two-devices}

We start with a simple network that includes a BS and two wireless users $c_1$ and $c_2$. The BS has \textit{a} common packet $q$ (i.e., common interest) that needs to be delivered to $c_1$ and $c_2$, respectively, as illustrated in Fig. \ref{fig:two-device}. This simplified scenario allows us to deliver clear analytic results and clean engineering insights. More general cases will be analyzed in the following sections; in particular, we will investigate \textit{dynamic packet arrivals} in Section \ref{section:dyanmic}.

\begin{figure}
\centering
\includegraphics[width=.4\textwidth]{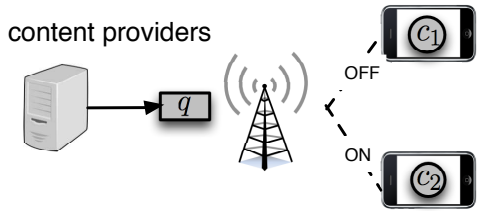}
\caption{Two users with a  common interest.}
\label{fig:two-device}
\end{figure}

We consider a discrete time system, in which the BS can transmit at most one packet in each time slot.  Transmissions  between the BS and each users are unreliable, where each B2D channel is describe as a time-varying ON/OFF channel\footnote{In this paper, we consider unreliable communications, i.e., either the BS does not use the rate-adaption to make error-free transmissions, or non-trivial B2D errors still occur with the rate-adaption scheme. As in many other papers, e.g., \cite{channel-model:i-hong,channel-model:yang}, our first focus is on the i.i.d. ON/OFF model to describe the reliability of wireless networks.}. By $\mathbf{s}(t)=(s_1(t),s_2(t))$,  we describe the channel state vector at time $t$, where $s_i(t) = 1$ if the BS can successfully transmit a packet to user $c_i$ at time $t$; otherwise, $s_i(t) = 0$. The channel state vector is assumed to be independent and identically distributed (i.i.d.) over time.  Moreover, we assume that the BS knows the channel state $\mathbf{s}(t)$ before transmitting a packet at time $t$, which can be achieved by existing channel estimation techniques (e.g. \cite{estimation}).       We will relaxing the assumptions in Section \ref{section:further-discussion} by considering correlated channels\footnote{The imperfect channel estimation and delayed feedback of the estimated channel states are beyond the scope of this paper.}.

In this section, we focus on the \textit{symmetric} users, where the channels have the same   error probability, i.e., $\mathbb{P}(s_i(t)=0)=p_e$ for all $t$ and $i$, while the \textit{asymmetric} case will be studied in Section \ref{section:asymmetric}. In this paper, we do not consider the mobility, i.e., the users are all staying in the same location for a certain amount of time; as such, assume a constant error probability in the following.

We are interested in the \textit{completion time} that is the expected time the BS takes to successfully send $q$ to all users.  Let  $T_{=}$ be the  completion time for \textit{broadcasting} the common interest to both users. We then can get\footnote{If we consider a  network consisting of only a BS and a user with its B2D channel error probability $p_e$,  the BS takes the expected time of $1/(1-p_e)$ to successfully deliver a packet to the user.}  
\begin{align}
T_{=}&=p_e^2(1+T_{=})+(1-p_e)^2 \cdot 1+2p_e(1-p_e)(1+\frac{1}{1-p_e}) \label{eq:Tc}\\
&=\frac{2p_e+1}{1-p_e^2}, \nonumber
\end{align}
where the second term in right-hand side (RHS) of Eq. (\ref{eq:Tc}) takes advantage of the broadcast medium when both channels are ON.



Here, we remark that though traditionally the users communicate with the BS via their cellular links \textit{independently} (i.e., the BS does not know if both users are requesting the same content; as such, can only \textit{unicast} their requested packets respectively), in Appendix \ref{appendix:b2d-identify} we will discuss an implementation of identifying the common interest; as such, the BS would deliver the common packet $q$ over the broadcast B2D. Moreover, we show that identifying the common interest under the unreliable networks will not that effective, as stated precisely in the following lemma.
\begin{lemma} \label{lemma:improve-common}
Let $T_{\neq}$ be the  completion time for the BS to deliver the common interest over the unicast B2D. The improvement ratio $R_{\neq/=}$ (defined below) from identifying the common interest is 
\begin{eqnarray*}
R_{\neq/=}:=\frac{T_{\neq}}{T_{=}}=\frac{p_e+2}{2p_e+1}. 
\end{eqnarray*} 
In particular, the ratio is 2 when $p_e=0$; the ratio is 1.25 when $p_e=0.5$. 
\end{lemma}
\begin{proof}
Please see Appendix \ref{appendix:b2d-identify}.
\end{proof}

We observe that the decrease of the improvement ratio $R_{\neq/=}$ in $p_e$ is faster than a linear function (i.e., a convex function). In particular, when the channels are terrible (i.e., $p_e \rightarrow 1$), the identification  almost cannot decrease the  completion time.

\subsection{Motivating social grouping} \label{subsection:incentive}

To further utilize the common interest, we take advantage of the existing D2D communication technology. In other words, we are considering the hybrid networks including B2D and D2D communications. We assume that D2D communications do not occupy the same channels as B2D communications. When a user is making a local transmission over a D2D network, the BS can save bandwidth to serve other users. Moreover, we reasonably consider the half-duplex technique,      i.e., each user is not allowed to use both B2D and D2D interfaces simultaneously, and cannot transmit and receive a packet at the same time.  Because of a short-distance communication, we  assume that the D2D channels between the users are noiseless,  while  imperfect D2D communications will be discussed in Subsection \ref{subsection:unreliable-local}.

However, the local users are  selfish or not cooperative in general, we  propose a non-monetary incentive mechanism,  called \textit{equal-reciprocal} incentive, which fairly requests the users to share the same expected amount of content with each other.  By means of the equal-reciprocal incentive scheme, the neighboring users $c_1$ and $c_2$ are motivated to form a social group.      Here, we assume that the users are  \textit{rational}, who can be motivated by the equal-reciprocal incentive and follow the incentive algorithm, but not malicious, who do not want to share the content. To exclude the malicious users, the users in the social group can exchange a key before any content transfer; as such, users outside the social group cannot decode received packets from the social group. Moreover, the BS can be the identity who isolates the malicious users from the social group, as in the paper the BS is responsible for regulating transmissions between the users, including who should be included in the group and who should share their own content at a time.

We note that the equal-reciprocal incentive for the two symmetric users is similar to the tit-for-tat scheme (based on unicast communications), while it will be quite different and efficient for scenarios of more than two users due to  broadcast D2D communications, which will be discussed later. Moreover, we will show that the equal-reciprocal incentive not only provides a fairness guarantee like the tit-fot-tat, it also maximizes a local utility assuring that a user can get the maximum amount of content from the social group  (see Subection \ref{subsection:individual}). 

%

To analyze the completion time subject to the equal-reciprocal incentive,  we consider a sharing probability $p_{i \rightarrow j}$ that $c_i$ shares the received packet with $c_j$ when only $c_i$ has got the packet.  Then, the expected number of packets sent from $c_i$ to $c_j$ is $(1-p_e)p_e p_{i \rightarrow j}$.  To motivate the social grouping with the equal-reciprocal constraint, we need $p_{1 \rightarrow 2}=p_{2\rightarrow 1}$. 

Let $T_{\cup}$ be the completion time to deliver the common interest over the broadcast B2D and D2D if $c_1$ and $c_2$ are incentivized to form the social group,  then
\begin{eqnarray*}
T_{\cup}&=&p_e^2(1+T_{\cup})+(1-p_e)^2\cdot 1+2(1-p_e)p_e\left(p_{1 \rightarrow 2}\cdot 2+(1-p_{1\rightarrow 2})(1+\frac{1}{1-p_e})\right),
\end{eqnarray*}    
when $p_{1 \rightarrow 2}=p_{2\rightarrow 1}$. To minimize the  completion time, we optimally choose $p^*_{1 \rightarrow 2}=p^*_{2\rightarrow 1}=1$, then we have
\begin{eqnarray*}
T^*_{\cup}=\frac{-2p_e^2+2p_e+1}{1-p_e^2}. \label{eq:Ts}
\end{eqnarray*} 

\begin{define}
An \textit{optimal equal-reciprocal} scheme is a sharing policy (i.e., sharing probabilities) that minimizes the completion time subject to the equal-reciprocal constraint. 
\end{define}

We notice that in this symmetric case, the optimal equal-reciprocal  scheme does not lose any performance compared to the \textit{full cooperation} (i.e., both users are always willing to help each other without an incentive scheme). Yet, this is not the case in the asymmetric case, as will be discussed in Section \ref{section:asymmetric}. 

\begin{lemma} \label{lemma:improve-from-social}
The improvement from the social grouping is 
\begin{align*}
R_{=/\cup}:=\frac{T_{=}}{T^*_{\cup}}=\frac{2p_e+1}{-2p_e^2+2p_e+1}.
\end{align*}
In particular, the ratio is 1 when $p_e=0$; the ratio is 1.33 when $p_e=0.5$. 
\end{lemma}

In Lemma  \ref{lemma:improve-from-social}, such an improvement is due to  the local sharing in the D2D network. With the aid of the incentivized social group,  the identified common interest can significantly reduce the completion time even in  terrible B2D communication environment, i.e., $p_e \rightarrow 1$ (the ratio  is 3  when $p_e=1$). 

%
%
%

\subsection{Large symmetric networks} \label{section:large-group}

In this subsection, we are extending the discussion to a large network consisting of a BS and $N$  wireless users $c_1, c_2, \cdots, c_N$. These $N$ users are within the proximity of each other, and desire a common interest. We consider the symmetric users with the same error probability $p_e$.

By $T^{(N)}_{=}$ we denote  the  completion time to deliver the common interest over the broadcast B2D. We present $T^{(N)}_{=}$ in recursion:  $T^{(0)}_{=}=0$, and for all $n=1, \cdots, N$,
\begin{eqnarray}
T^{(n)}_{=}=\sum^{n}_{i=0} {n \choose i} p_e^i(1-p_e)^{N-i}(1+T^{(i)}_{=}), \label{eq:Tnc}
\end{eqnarray}
where $1+T^{(i)}_{=}$ is the expected number of transmissions under condition that $i$ users do not get the packet from the BS's first broadcast.  We show $T^{(N)}_{=}$ in Fig. \ref{fig:Tnc}, where $T^{(N)}_{=}$ increases in $N$ and depends on $p_e$.

\begin{figure}
\centering
\includegraphics[width=.5\textwidth]{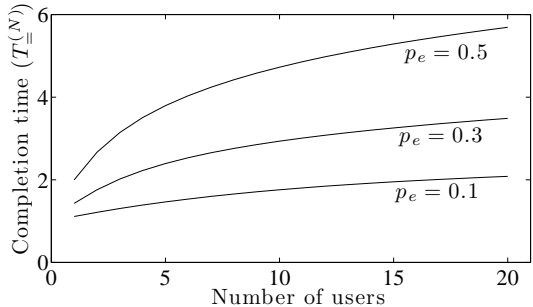}
\caption{Completion time $T^{(N)}_{=}$ (see Eq. (\ref{eq:Tnc})) for large symmetric networks with common interest identification and  no  social grouping.}
\label{fig:Tnc}
\end{figure}

We further exploit the local sharing benefit among users by means of broadcast D2D. We generalize the idea of the equal-reciprocal mechanism  by ensuring that \textit{each pair} in the social group share the same expected amount of content with each other. Due to the broadcast D2D channels among the local users, the analysis of the equal-reciprocal incentive will be different from the tit-for-tat incentive for P2P networks (based on unicast communications).

Let $\mathbf{R}$ be the set of \textit{remaining} users that do not get the packet after the first B2D transmission. 
By $p_{i \rightarrow \mathbf{R}}$ we denote a sharing probability that user $c_i$ shares the packet with  $\mathbf{R}$ when $c_i$ has got the packet. 

\begin{lemma}
The optimal equal reciprocal scheme for large symmetric networks is as follows: to pick a user who is responsible for the broadcast with the equal probability; in particular, we select
\begin{eqnarray}
p^*_{i \rightarrow \mathbf{R}}=\frac{1}{N-|\mathbf{R}|}.\label{eq:large-network}
\end{eqnarray}
\end{lemma}
\begin{proof}
The policy  incentivizes all $N$ users to participate in the social group as the equal-reciprocal constraint holds for all pairs of the local users. Moreover, with the choice of the sharing probability, if there exists a user that has got the packet,  all clients can recover the packet in the next time slot as $\sum_{i \notin \mathbf{R} } p^*_{i \rightarrow \mathbf{R}}=1$ according to Eq. (\ref{eq:large-network}).
\end{proof}

Let $T^{(N)}_{\cup}$ be the  completion time  to deliver the common interest over the broadcast B2D and D2D, associated with the sharing probability in Eq. (\ref{eq:large-network}). Then, we have
\begin{align}
T^{(N)}_{\cup}=&p_e^N(1+T^{(N)}_{\cup})+(1-p_e)^N+2\left(1-p_e^N-(1-p_e)^N\right) \label{eq:Ts-N}\\
=&1+\frac{1-(1-p_e)^N}{1-p_e^N}, \nonumber
\end{align}
where the third term in RHS of Eq. (\ref{eq:Ts-N}) reflects the local content sharing in the group, i.e., one of the  users that has received the packet in the first time slot will broadcast it such that all users in the group can recover the packet.  

Note that $T^{(N)}_{\cup}$ is monotonic in $N$, but can be  increasing or decreasing due to the tradeoff:  
\begin{itemize}
	\item \emph{More users to serve:} As the number of local users increases, the probability that all B2D channels  are ON decreases;
	\item \emph{User diversity increases:} As the number of local users increases, the probability that one user receives the packet at the first transmission increases. 
\end{itemize}

\begin{lemma} \label{theorem:large-symmetric}
When $p_e<0.5$, $T^{(N)}_{\cup}$ increases in $N$. When $p_e>0.5$,  $T^{(N)}_{\cup}$ decreases in $N$. When $p_e=0.5$, $T^{(N)}_{\cup}=2$ for all  $N$. Moreover, $T^{(N)}_{\cup}$ approaches 2 when $N \rightarrow \infty$, independent of $p_e$, while $T^{(N)}_{=}$ is sensitive to $p_e$ as shown in Fig. \ref{fig:Tnc}.  
\end{lemma}

When the B2D channels are quite unreliable (i.e., $p_e > 0.5$), a larger social group can better use the user diversity advantage and shorten the  completion time.

\section{General asymmetric networks} \label{section:asymmetric}
Thus far,  the optimal equal-reciprocal scheme can satisfy all symmetric users in the next time slot once one user has got the packet.  In this section, we analyze the optimal equal-reciprocal  scheme for asymmetric users.

\subsection{Two asymmetric users}
We first reconsider the simple network in Fig. \ref{fig:two-device}, but the channel between the BS and each user $c_i$ has an error probability $p_{e,i}$. Without loss of generality, we assume that $p_{e,1} < p_{e,2}$. 

\subsubsection{No content sharing}
If $c_1$ and $c_2$ do not share any content,  the  completion time for delivering the common interest over broadcast B2D is
\begin{eqnarray} 
T_{=}&=&\frac{1}{1-p_{e,1}p_{e,2}}\Bigl[p_{e,1}p_{e,2}+(1-p_{e,1})(1-p_{e,2})+ \nonumber\\
&&p_{e,1}(1-p_{e,2})\left(1+\frac{1}{1-p_{e,1}}\right)+ p_{e,2}(1-p_{e,1})\left(1+\frac{1}{1-p_{e,2}}\right)\Bigr]. \label{eq:no-info-share}
\end{eqnarray}

\subsubsection{Full cooperation}
We consider that $c_1$ and $c_2$ are  \textit{always} help each other. We refer to this case as a  \textit{full cooperation}. We  use this case as a benchmark for later comparison purpose with another case. Then, the  completion time $T_{f}$ is
\begin{eqnarray*}
T_f&=&\frac{1}{1-p_{e,1}p_{e,2}}\Bigl[p_{e,1}p_{e,2}+(1-p_{e,1})(1-p_{e,2})+ p_{e,1}(1-p_{e,2})\cdot 2+ p_{e,2}(1-p_{e,1})\cdot 2\Bigr],
\end{eqnarray*}
where the last two terms indicate the full cooperation: if either one receives the packet, it will share with the other one, resulting in totally two time slots. 

With the full cooperation, the expected amount of packets sent from $c_i$ to $c_j$ is $(1-p_{e,i})p_{e,j}p_{i \rightarrow j}$. As $p_{e,1} < p_{e,2}$, if $c_1$ and $c_2$ always help  each other, i.e., $p_{1 \rightarrow 2}=p_{2 \rightarrow 1}=1$, then 
$$(1-p_{e,2})p_{e,1}p_{2 \rightarrow 1} < (1-p_{e,1})p_{e,2}p_{1 \rightarrow 2}.$$

However, if $c_1$ and $c_2$ can be selfish in general, the full cooperation does not result in a fair situation as $c_1$ needs to share more packets with $c_2$. Thus, it is difficult for the selfish $c_1$ and $c_2$ to  establish the full cooperation.

\subsubsection{Equal-reciprocal incentive}
To motivate the selfish $c_1$ and $c_2$ to form a social group, we apply the proposed equal-reciprocal incentive. Using the equal-reciprocal constraint as an incentive, we  fairly choose a unique sharing probability $p_{i \rightarrow j}$ such that
$$(1-p_{e,2})p_{e,1}p_{2 \rightarrow 1} = (1-p_{e,1})p_{e,2}p_{1 \rightarrow 2}.$$
To minimize the completion time, we optimally select 
\begin{eqnarray}
p^*_{1 \rightarrow 2} = \frac{(1-p_{e,2})p_{e,1}}{(1-p_{e,1})p_{e,2}} < 1;  \hspace{.5cm} p^*_{2 \rightarrow 1} = 1. \label{eq:p12}
\end{eqnarray}
The intuition is that $c_2$ is associated with a worse B2D channel, so in  the long-run  it  get fewer packets from the BS than $c_1$. Hence, to minimize the completion time, it is the best that $c_2$ always share  content, while $c_1$ share partial content with the same amount as $c_2$.  
Then, the optimal completion time $T^*_{\cup}$ in the presence of the incentivized social group is
\begin{align}
T^*_{\cup}=&\frac{1}{1-p_{e,1}p_{e,2}}\Bigl\{p_{e,1}p_{e,2}+(1-p_{e,1})(1-p_{e,2})+ (1-p_{e,2})p_{e,1}\cdot 2\nonumber\\
&+ (1-p_{e,1})p_{e,2}\cdot \Bigl[\frac{(1-p_{e,2})p_{e,1}}{(1-p_{e,1})p_{e,2}}\cdot 2+ \frac{p_{e,2}-p_{e,1}}{(1-p_{e,1})p_{e,2}}(1+\frac{1}{1-p_{e,2}})\Bigr]\Bigr\}.\label{eq:Ts-nonsymmetric}
\end{align}

By comparing $T^*_{\cup}$ and $T_f$, we notice the performance loss from the benchmark case due to the  user asymmetry and  selfishness. We then conclude in the following lemma.  

\begin{lemma}
Under our optimal equal-reciprocal scheme, to motivate the asymmetric $c_1$ and $c_2$ to form the social group has a shorter  completion time, i.e., $T^*_{\cup} \leq T_{=}$. Moreover,  the incentivized social group has a performance loss compared with the full cooperation,  i.e.,  
\begin{align*}
T^*_{\cup}-T_f=\frac{\Delta}{1-(p_{e,2}-\Delta)p_{e,2}}(\frac{1}{1-p_{e,2}}-1),
\end{align*}
where $\Delta=p_{e,2}-p_{e,1}$.
\end{lemma}
Taking the derivative to the above equation, we get  
\begin{align*}
\frac{d}{d\Delta}(T^*_{\cup}-T_f)=\frac{1-p_{e,2}^2}{(p_{e,2}\Delta+1-p_{e,2}^2)^2},
\end{align*}
and remark that 
\begin{itemize}
	\item $\frac{d}{d\Delta}(T^*_{\cup}-T_f) >0$; hence, $T^*_{\cup}-T_f$ increases in $\Delta$; 
	\item the rate $\frac{d}{d\Delta}(T^*_{\cup}-T_f)$ decreases as the difference $\Delta$ increases. 
\end{itemize}

\subsection{Large asymmetric networks} \label{subsection:large-asymmetric}

%
%

We focus on a network consisting of three users $c_1$, $c_2$, and $c_3$ with the channel error probabilities $p_{e,1}$, $p_{e,2}$, and $p_{e,3}$, respectively, while a more general network can be easily extended. Without loss of generality, we assume that $p_{e,1} \leq p_{e,2} \leq p_{e,3}$.

We remind that $p_{i \rightarrow \mathbf{R}}$ is the sharing probability of $c_i$ with a user set $\mathbf{R}$. According to following eight cases associated with all possible channel states $\mathbf{s}(t)=(s_1(t), s_2(t), s_3(t))$ at time $t=1$,  we work out  the  completion time $T_{\cup}$ for delivering the common interest with social grouping of $c_1, c_2, c_3$. Let $T_i$ be the completion time for  case $i$.

\begin{enumerate}

\item If $\mathbf{s}(1)=(0,1,1)$: the completion time  $T_1=2$, where we select $p_{2 \rightarrow 1}$ and $p_{3 \rightarrow 1}$ such that $p_{2 \rightarrow 1}+p_{3 \rightarrow 1}=1$. Similar to Eq. (\ref{eq:p12}), the selection is optimal.

\item  If $\mathbf{s}(1)=(1,0,1)$: the completion time  $T_2=2$, where we select that $p_{1 \rightarrow 2}+p_{3 \rightarrow 2}=1$. 

\item  If $\mathbf{s}(1)=(1,1,0)$: the  completion time  
$$T_3=(p_{1 \rightarrow 3}+p_{2 \rightarrow 3}) \cdot 2+ (1-p_{1 \rightarrow 3}-p_{2 \rightarrow 3})(1+\frac{1}{1-p_{e,3}}),$$ 
where the first term results from the content sharing in the social group, while the second one is due to the re-transmissions from the BS.

\item  If $\mathbf{s}(1)=(0,0,1)$: the  completion time  $T_4=2$, where we select that $p_{3 \rightarrow 1,2}=1$ due to the worst channel condition.  

 \begin{figure*}[!b]
 \hrulefill
\begin{align}
&T_5=p_{2 \rightarrow 1,3} \cdot 2+ (1-p_{2 \rightarrow 1,3})\Biggl[1+ \frac{1}{1-p_{e,1}p_{e,3}} \Bigl[ p_{e,1}p_{e,3}+(1-p_{e,1})(1-p_{e,3})+2p_{e,1}(1-p_{e,3})+\nonumber\\
&\hspace{5cm}  p_{e,3}(1-p_{e,1})\bigl[\frac{2(1-p_{e,3})p_{e,1}}{(1-p_{e,1})p_{e,3}}+ \frac{p_{e,3}-p_{e,1}}{(1-p_{e,1})p_{e,3}}(1+\frac{1}{1-p_{e,3}})\bigr]\Bigr]\Biggr].
\label{eq:social-case5}
\end{align}
\end{figure*}

\item If $\mathbf{s}(1)=(0,1,0)$: the  completion time is given in Eq. (\ref{eq:social-case5}), where the first term implies that $c_2$ broadcasts the packet to $c_1$ and $c_3$, while the remaining terms are because $c_2$ does not  share the content but $c_1$ and $c_3$ help  each other with some probability as in Eq. (\ref{eq:Ts-nonsymmetric}).

\item If $\mathbf{s}(1)=(1,0,0)$: the  completion time $T_6$ is similar to Eq. (\ref{eq:social-case5}). 

\item If $\mathbf{s}(1)=(1,1,1)$: the  completion time  $T_7=1$.

\item  If $\mathbf{s}(1)=(0,0,0)$: the  completion time  $T_8=T_{\cup}$.

\end{enumerate}

It is easy to evaluate the probability $\mathbb{P}(\text{case\,\,} i)$ of case $i$; then, the completion time is
\begin{align}
T_{\cup}=\frac{\mathbb{P}(\text{case\,\,} 8)+\sum_{i=1}^7 \mathbb{P}(\text{case\,\,} i) T_i}{1-p_{e,1}p_{e,2}p_{e,3}}. \label{eq:three-node-t}
\end{align}

We note that the  completion time $T_{\cup}$ is a linear function of the sharing probabilities, and hence conclude as follows.  

\begin{theorem}\label{theorem:linear-program}
We can get the optimal  completion time $T^*_{\cup}$ and  the optimal equal-reciprocal scheme (i.e., the optimal sharing probabilities) by formulating a \textit{linear program} as follows. 

\textbf{Linear program}:
\begin{align*}
&\min \frac{\mathbb{P}(\text{case\,\,} 8)+\sum_{i=1}^7 \mathbb{P}(\text{case\,\,} i) T_i}{1-p_{e,1}p_{e,2}p_{e,3}}\\
&\text{Subject to the following constraints:} 
\end{align*} 

(\textit{\textbf{Const. 1}}) Equal-reciprocal constraint between $c_1$ and $c_2$:
\begin{eqnarray*}
&&(1-p_{e,1})p_{e,2}\left[p_{e,3}\cdot p_{1 \rightarrow 2,3}+(1-p_{e,3}) p_{1 \rightarrow 2}\right] = p_{e,1}(1-p_{e,2})\left[p_{e,3}\cdot p_{2 \rightarrow 1,3}+(1-p_{e,3}) p_{2 \rightarrow 1}\right].
\end{eqnarray*}

(\textit{\textbf{Const. 2}}) Equal-reciprocal constraint between $c_1$ and $c_3$:
\begin{eqnarray*}
(1-p_{e,1})p_{e,3}\left[p_{e,2}\cdot p_{1 \rightarrow 2,3}+(1-p_{e,2}) p_{1 \rightarrow 3}\right] = p_{e,1}(1-p_{e,3})\left[p_{e,2}\cdot p_{3 \rightarrow 1,2}+(1-p_{e,2}) p_{3 \rightarrow 1}\right].
\end{eqnarray*}

(\textit{\textbf{Const. 3}}) Equal-reciprocal constraint between $c_2$ and $c_3$:
\begin{eqnarray*}
(1-p_{e,2})p_{e,3}\left[p_{e,1}\cdot p_{2 \rightarrow 1,3}+(1-p_{e,1}) p_{2 \rightarrow3}\right] = p_{e,3}(1-p_{e,2})\left[p_{e,1}\cdot p_{3 \rightarrow 1,2}+(1-p_{e,1}) p_{3 \rightarrow 2}\right].
\end{eqnarray*}

(\textit{\textbf{Const. 4}}) For the cases 1, 2, and 3, the total sharing probability to $c_i$ is no more than one for all $i$:
\begin{align*}
 p_{2 \rightarrow 1}+p_{3 \rightarrow 1} = 1; \hspace{1cm} p_{1\rightarrow 2}+p_{3 \rightarrow 2} =1; \hspace{1cm} p_{1 \rightarrow 3}+p_{2 \rightarrow 3} \leq 1. 
\end{align*}

(\textit{\textbf{Const. 5}}) All sharing probabilities are less than or equal to one.
\end{theorem}

We remark that the linearity of the  completion time still holds for  $N$ users; as such, the optimal  completion time for the $N$ users can be derived, though complicated, by recursively solving the linear program for $i$ users (poly-time solvable), where $1 \leq i \leq N$. We want to emphasize that  Sections \ref{section:two-devices} and \ref{section:asymmetric} work on the performance analysis of the equal-reciprocal incentive for a common interest, whereas we will propose on-line algorithms selecting a user to share at a time, without the complex computation of the sharing probabilities.

\begin{figure}
\centering
\includegraphics[width=.5\textwidth]{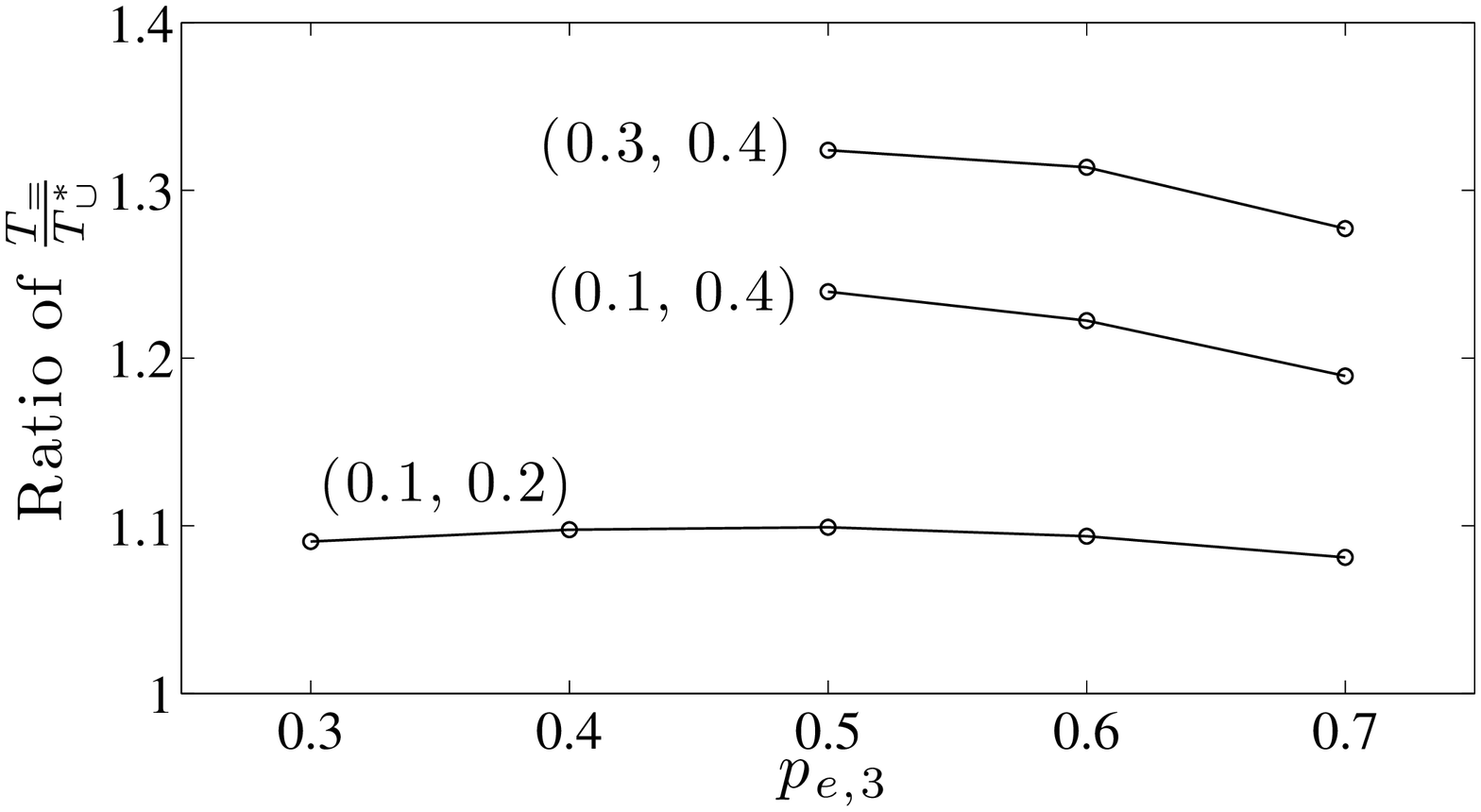}
\caption{The ratio of $T_{=}/T^*_{\cup}$ for various pairs of $(p_{e,1}, p_{e,2})$.}
\label{fig:three-node-result}
\end{figure}

As  there is no closed-form solution for the linear program, we  numerically demonstrate the improvement ratio of $T_{=}/T^*_{\cup}$ in Fig. \ref{fig:three-node-result}, where  each curve is  for a particular pair of $(p_{e,1}, p_{e,2})$ and starts with $p_{e,3} > p_{e,2}$. 

Thus far, the sharing probabilities are chosen to motivate \textit{all} neighboring users to join in the social group. However, we have understood that adding a user with a poor B2D channel will reduce sharing opportunity in the group. Hence, \textit{to incentivize a user or not} is a question.  We conclude as follows.

\begin{theorem}
To incentivize all neighboring users to form a social group minimizes the completion time, no matter what the distribution of B2D channel errors is. 
\end{theorem}
\begin{proof}
See Appendix \ref{section:to-incentivize}.
\end{proof}

%

%
%
%

\section{On-line scheduling for dynamic content arrivals} \label{section:dyanmic}

\begin{figure}
\centering
\includegraphics[width=.4\textwidth]{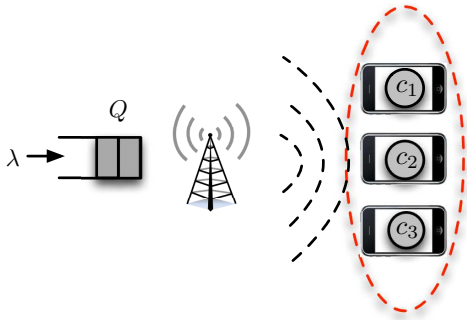}
\caption{Example queueing network}
\label{fig:dyanmics}
\end{figure}

As we considered single common interest before,  we did not need to schedule transmissions for lots of packets. 
In this section, we further consider  time-varying packet arrivals at  the BS, and aim at designing \textit{on-line transmission scheduling} algorithms that decide who (i.e., the BS or users) should transmit content for each time. Based on the algorithm, the BS will adaptively select a user to share received content (if needed) such that the equal-reciprocal incentive is met. As the algorithm is on-line fashion, our incentive design does not need to pre-compute the sharing probabilities in advance and would be effective to keep communications in timely fashion. Moreover, by simply using additional one bit to indicate sharing or not, the users can be informed. In summary, our algorithm is efficient in both time and space.

Different from the discussions before, networks in this section include \textit{queues} in the BS to store content that will be delivered to users.  For a clear explanation,  we present our design in the context of a queueing network in Fig. \ref{fig:dyanmics}, where we consider three asymmetric users. We remark that our design can be easily extended for any number of users, which will be discussed later. We consider the ON/OFF B2D channels and the perfect D2D channels first, while correlated B2D channels and imperfect D2D transmissions are discussed in Section \ref{section:further-discussion}.

In addition to the time-varying channels, we  start to consider  time-varying packet arrivals, where  packets in a queue $Q$ are  common interests for $c_1$, $c_2$, and $c_3$. The arrival process is assumed to be i.i.d. over time.   Let $\lambda$ be the packet arrival rate to the queue $Q$.  We say that the queue is \textit{stable} for $\lambda$ (see \cite{neely:book}) if the time average queue length is finite when the arrival rate at the BS is $\lambda$.

\begin{figure*}[!t]
\centering
\includegraphics[width=.8\textwidth]{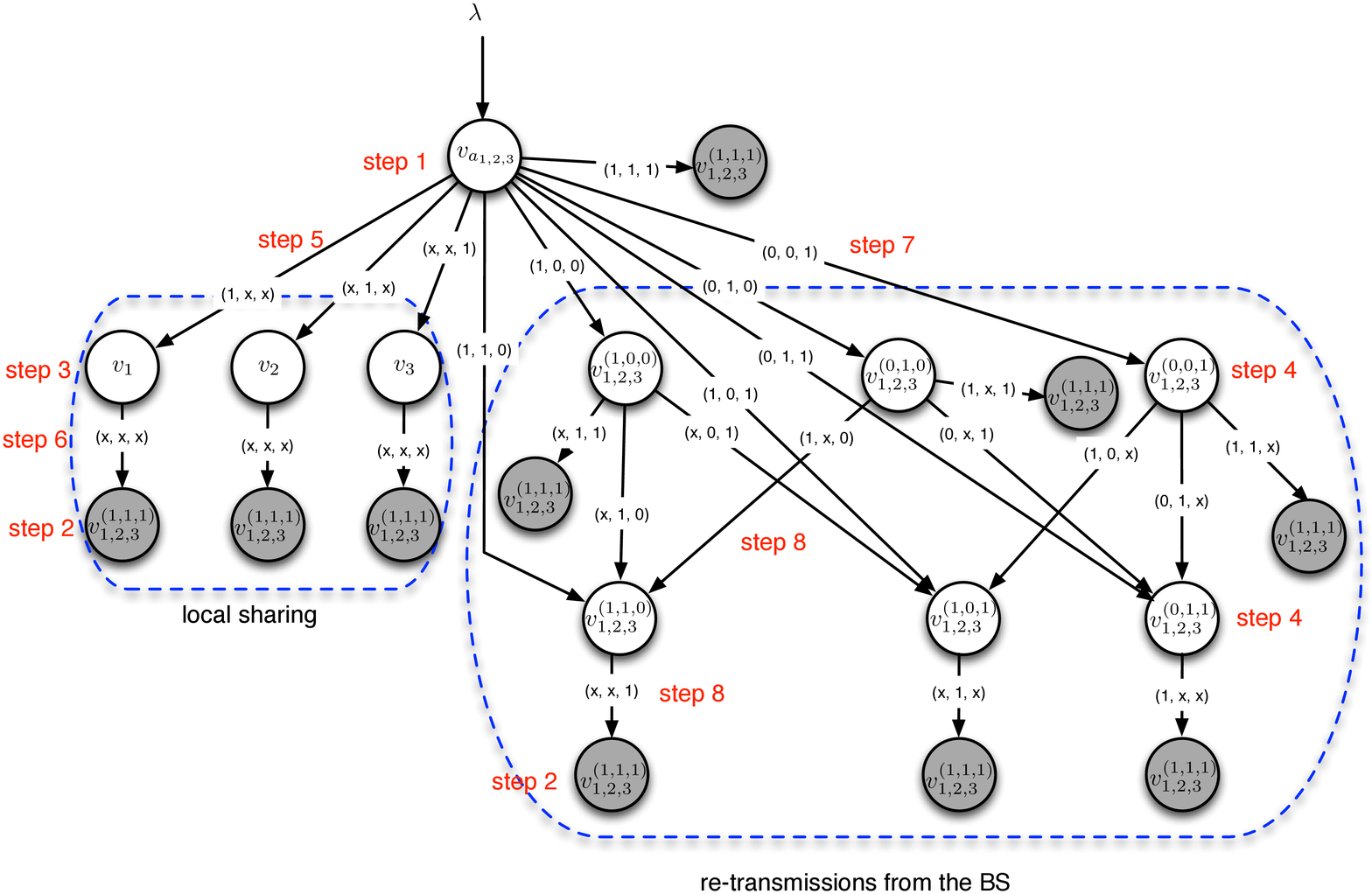}
\caption{Virtual network of Fig. \ref{fig:dyanmics}, where the gray nodes are the destination nodes.}
\label{fig:virtual}
\end{figure*}

A \textit{stability region}  consists of all arrival rates $\lambda$ such that there exists a transmission schedule ensuring that the queue is stable. A scheduling algorithm is \textit{throughput-optimal} if the queue is stable for all arrival rates  interior of the stability region. Moreover,  by $\mathbf{\Lambda}$ we denote an \textit{equal-reciprocal stability region} that is the stability region subject to the equal-reciprocal constraint. 
We then re-define an \textit{optimal equal-reciprocal scheme} for this dynamic case as a scheduling algorithm that stabilizes the queue and meets the equal-reciprocal constraint for all arrival rates  interior of $\mathbf{\Lambda}$. In other words, the optimal equal-reciprocal scheme is  throughput-optimal subject to the equal-reciprocal incentive. 
  Because of the time-varying channels and packet arrivals,  scheduling algorithms will depend on the channel state, rate, and queue size.

\subsection{Optimal centralized scheduling  under the equal-reciprocal incentive}

As the network in Fig. \ref{fig:dyanmics} includes the multiple multicasts, the first step of our scheduling design is to create a \textit{virtual network} as shown in Fig. \ref{fig:virtual}. In general, we are innovating on the lossy broadcast network by applying prior  theory about wireline networks (e.g., \cite{neely:book}).  We discuss the idea of the virtual network in the following ten steps.

\begin{enumerate}

	\item \textit{Create a virtual node $v_{a_{1,2,3}}$ corresponding to the real queue $Q$} : If a real packet arrives at $Q$, a virtual packet is created in $v_{a_{1,2,3}}$.
	
	\item \textit{Create a virtual node $v_{1,2,3}^{(1,1,1)}$ that has the empty queue}: The virtual node is the destination of the virtual packets in $v_{a_{1,2,3}}$. If a virtual packet arrives at  $v_{a_{1,2,3}}^{(1,1,1)}$, the corresponding real packet is  successfully delivered. 
	
	\item \textit{Create a virtual node $v_i$ related to  user $c_i$ for $i=1,2,3$}: A virtual packet in $v_i$ implies that user $c_i$ should store the packet and share it when scheduled. 
	
	\item \textit{Create a virtual node $v_{1,2,3}^{(r_1,r_2,r_3)}$ for $r_1,r_2,r_3= 0, 1$}: The virtual nodes are associated with the re-transmission status of the BS. The superscript $r_i$ indicates if user $c_i$ has received the packet or not, where if $r_i=1$ the packet has reached at $c_i$; else not yet. For example, a virtual packet in the node $v_{1,2,3}^{(1,1,0)}$ means that the corresponding packet in the real network has been received by $c_1$ and $c_2$, but not by $c_3$. 
	

	\item \textit{Create a virtual link $v_{a_{1,2,3}} \rightarrow v_i$,  for $i=1,2,3$, associated with an ON/OFF channel}: The virtual link is ON at time $t$ when $s_i(t)=1$ and $s_j(t)=x$ for all $j \neq i$, where $x$ can be either 0 or 1; else, the virtual link is OFF.  

	\item \textit{Create a virtual link $v_{i} \rightarrow v_{1,2,3}^{(1,1,1)}$ associated with a noiseless channel  for $i=1,2,3$}: The virtual links are always ON because of the noiseless D2D channels among the users $c_1$, $c_1$, and $c_3$.

	\item \textit{Create a virtual link $v_{a_{1,2,3}} \rightarrow v_{1,2,3}^{(r_1,r_2,r_3)}$ associated with an ON/OFF channel for $r_1, r_2, r_3=0,1$}: The virtual link is ON at time $t$ when $s_i(t)=r_i$, for $i=1,2,3$. 

	\item \textit{Create a virtual link $v_{1,2,3}^{r_1, r_2, r_3} \rightarrow v_{1,2,3}^{(r_1^{'},r_2^{'},r_3^{'})}$  associated with an ON/OFF channel for $r_1, r_2, r_3=0,1$ and $r_i^{'} \geq r_i$ for all $i$}: The virtual link is ON at time $t$ when (\textit{i}) $s_i(t)=x$ if $r_i=1$, for $i=1,2,3$ (i.e., if $c_i$ has received the packet, re-transmitting the packet does not need to consider the B2D channel to $c_i$); (\textit{ii})  $s_i(t)=r_i^{'}$ if $r_i=0$, for $i=1,2,3$.

	\item \textit{All virtual links cause interference to each other}: That implies that in the real network a user cannot simultaneously transmit and receive a packet  and cannot use B2D and D2D interfaces at the same time. 
	
		\item Except for the virtual node $v_{1,2,3}^{(1,1,1)}$, other virtual nodes have their own  virtual queues.   

\end{enumerate}

\begin{algorithm}[!t]
\small
\SetCommentSty{text}
\SetAlgoLined 
\SetKwFunction{Union}{Union}\SetKwFunction{FindCompress}{FindCompress} \SetKwInOut{Input}{input}\SetKwInOut{Output}{output}

\Input{A network instance of Fig. \ref{fig:dyanmics}}
\Output{Scheduling decision for each time}

\tcc{\textit{\textbf{Virtual network construction}}}

Construct a virtual network as shown in Fig. \ref{fig:virtual}\; \label{alg1:virtual-network}

Add other virtual queues $h_{i,j}$ to the virtual network, for $i=1,2$ and $i+1 \leq j \leq 3$\; \label{alg1:virtual-queue}

\tcc{\textit{\textbf{Virtual link scheduling for each time $t$}}}

For each virtual link $l = u \rightarrow d$ in Fig. \ref{fig:virtual}, we define a weight $W_{l}(t)$  and  a variable $\mu_l(t)$ as  \label{alg1:weight1}
\begin{eqnarray*}
W_{l}(t)=\max\{0,V_u(t)-V_d(t)\};
\end{eqnarray*}
\begin{equation*}
\mu_l(t)=
\left\{
\begin{array}{ll}
1 & \text{if virtual link $l$ is ON and scheduled;}\\
0 & \text{else;}
\end{array}
\right.
\end{equation*}

Define a variable $n_{i,j}(t)$ for $i,j=1,2,3$ as \label{alg1:weight2}
\begin{equation*}
n_{i,j}(t)=
\left\{
\begin{array}{ll}
1 & \text{if virtual link $v_{a_{1,2,3}} \rightarrow v_i$ is scheduled and $s_j(t)=0$;}\\
0 & \text{else;}
\end{array}
\right.
\end{equation*}

Schedule the virtual link  $l^* \in \mathbf{L}$ that maximizes \label{alg1:backpressure}
\begin{equation}
\sum_{l \in L} W_{l}(t) \mu_{l}(t)+\sum^2_{i=1} \sum^3_{j=i+1} H_{i,j}(t)(n_{j,i}(t)-n_{i,j}(t));
\label{eq:back-pressure}
\end{equation}

\tcc{\textit{\textbf{Real packet scheduling for each time $t$}}}

\If{
$v_{a_{1,2,3}} \rightarrow v_i$, for $i=1,2,3$, is scheduled \label{alg1:schedule-start}
}{
the BS transmits the corresponding packet, while $c_i$  needs to store the packet and broadcast it when scheduled in the future\;
}

\If{$v_{i} \rightarrow v_{1,2,3}^{(1,1,1)}$ is scheduled}{ 
user $c_i$ is scheduled to share the corresponding packet\;}

\If{$v_{a_{1,2,3}} \rightarrow v_{1,2,3}^{(r_1,r_2,r_3)}$ or $v_{1,2,3}^{(r_1, r_2, r_3)} \rightarrow v_{1,2,3}^{(r_1^{'},r_2^{'},r_3^{'})}$, for $r_i=0,1$, $r^{'}_i \geq r_i$, and $i=1,2,3$,  is scheduled}{
the BS broadcasts the corresponding packet\; 
}\label{alg1:schedule-end}

\tcc{\textit{\textbf{Virtual queueing update}}}

\ForEach{virtual link $u \rightarrow d$ that is scheduled \label{alg1:update}}{
$V_{u}(t+1)=\max\{0, V_{u}(t)-1\}$\; \label{alg1:u-update}
$V_{d}(t+1)=V_d(t)+1$\; \label{alg1:d-update}
}

\If{there is an arrival to $v_{a_{1,2,3}}$}{
Increase its virtual queue size by 1\;}

$H_{i,j}(t+1)=H_{i,j}(t)+n_{i,j}(t)-n_{j,i}(t)$ for $i=1,2$, and $i+1 \leq j \leq 3$\;  \label{alg1:h-update}

\caption{Optimal on-line scheduling algorithm}
\label{alg:centralized}
\end{algorithm}

Let $\mathbf{V}$ be the set of all vertices in Fig. \ref{fig:virtual}. For each time $t$, if a virtual link $u \rightarrow d$ is scheduled, a virtual packet is delivered from an upstream virtual node $u \in \mathbf{V}$  to a downstream virtual node $d \in \mathbf{V}$; meanwhile,   their virtual queue sizes are updated: $V_u(t+1)=\max\{0, V_u(t)-1\}$ and $V_d(t+1)=V_d(t)+1$ (see Lines \ref{alg1:u-update} and \ref{alg1:d-update} of Alg. \ref{alg:centralized} as introduced soon).

Now, we are ready to propose an optimal equal-reciprocal scheme in Alg. \ref{alg:centralized}. The proposed algorithm has four important parts: (1) virtual network construction; (2) virtual link scheduling; (3) real packet scheduling;  (4) queueing update.

We notice that transmission schedules in the real network are associated with the virtual link schedules, as described in Lines \ref{alg1:schedule-start}-\ref{alg1:schedule-end}. With the relationship between the real and virtual networks, the real network is stable if and only if the virtual network is stable. Moreover, a more general result is stated in the following lemma. 
\begin{lemma}
The stability region of the virtual network (in Fig. \ref{fig:virtual}) is the same as that of the real network (in Fig. \ref{fig:dyanmics}).
\end{lemma}

We remark that the virtual network can be divided into two parts (see Fig. \ref{fig:virtual}), where the left one is related to the \textit{local content sharing}, while the right one is corresponding to \textit{re-transmissions} from the BS. Then, the optimal equal-reciprocal scheme  will schedule the virtual link such that the virtual queue is stable and  the local sharing meets the equal-reciprocal constraint, for all packet arrival rates interior of $\mathbf{\Lambda}$.

Motivated by \cite{p2p:neely}, in Line \ref{alg1:virtual-queue} of Alg. \ref{alg:centralized} we introduce other virtual queues $h_{i,j}$ for $i=1,2$ and $i+1 \leq j \leq 3$, with the virtual queue size $H_{i,j}(t)$ at time $t$. By $n_{i,j}(t)$ we indicate if at time $t$ the BS broadcasts a packet to $c_i$ who will need to store and share the packet with $c_j$ when scheduled in the future. That is, if  virtual link $v_{a_{1,2,3}} \rightarrow v_i$ is scheduled at time $t$ and the channel state $s_j(t)=0$, then $n_{i,j}(t)=1$. We describe the queueing dynamics of $h_{i,j}$  in Line \ref{alg1:h-update} of Alg. \ref{alg:centralized}, and have the following result. 
\begin{lemma}
If $h_{i,j}$ for all $i=1,2$ and $i+1 \leq j \leq 3$ are stable, then the equal-reciprocal  constraint is met. 
\end{lemma}
\begin{proof}
According to the queueing dynamics of $h_{i,j}$, we get
\begin{align*}
H_{i,j}(t)=H_{i,j}(0)+\sum^{t}_{\tau=0} n_{i,j}(t)-n_{j,i}(t).
\end{align*}
Dividing both sides above by $t$ yields
\begin{align*}
\frac{H_{i,j}(t)}{t}=\frac{\sum^{t}_{\tau=0} n_{i,j}(t)-n_{j,i}(t)}{t}.
\end{align*}
Let $t \rightarrow \infty$, and we know that  the number of packets shared between $c_i$ and $c_j$ is equal if the virtual queue $h_{i,j}$ is stable. 
\end{proof}

Thus, the optimal equal-reciprocal scheme further becomes to schedule the virtual links such that all  virtual queues, including the virtual queues in Fig. \ref{fig:virtual} and $h_{i,j}$ for all $i=1,2$ and $i+1 \leq j \leq 3$,  are stable for all arrival rates interior of $\mathbf{\Lambda}$.

To that end, we apply a back-pressure type algorithm \cite{p2p:neely,neely:book}.  
Let $\mathbf{L}$ be a set of all the virtual links. 
Hence, we optimally schedule the virtual link set $l^* \in \mathbf{L}$ according to Line \ref{alg1:backpressure} with some weights defined in Lines \ref{alg1:weight1} and \ref{alg1:weight2}, which are related to the number of served packets if a virtual link is scheduled.

Because the virtual link schedule in Alg. \ref{alg:centralized} can stabilize all virtual queues and satisfy the equal-reciprocal constraint for all arrival rates interior of $\mathbf{\Lambda}$,  we conclude as follows. 
\begin{theorem}
The proposed on-line scheduling algorithm in Alg. \ref{alg:centralized} is an optimal equal-reciprocal scheme. 
\end{theorem}

 Our scheduling design can be  extended to any number of neighboring users. First, each link $v_{a_{1, \cdots}} \rightarrow v_i$ in the local sharing part is associated with an ON/OFF channel, which is ON at time $t$ when $s_i(t)=1$. Second,  each link in the re-transmission part is associated with another ON/OFF channel according to the users who have not received the packet yet (similar to the steps 7 and 8 in Subsection IV-A). Based on the virtual network, the proposed Alg. \ref{alg:centralized} can be easily generalized.

\subsection{Sub-optimal distributed  scheduling  under the equal-reciprocal incentive}
In Alg. \ref{alg:centralized},  transmissions within the social group are scheduled by the BS, so the BS needs to inform a user to broadcast a packet by sending an additional message. To reduce these control overhead,  we hence develop a distributed  algorithm in this subsection. The idea is that if a virtual link $v_{a_{1,2,3}} \rightarrow v_{i}$ is scheduled at present, the virtual link $v_i \rightarrow v_{1,2,3}^{(1,1,1)}$ needs to be scheduled in the next time slot. In this way,  broadcasts in the social group do not need to be scheduled by the BS. Hence, the  virtual link schedule is made at the beginning of each time, except for the time slot after  virtual link $v_{a_{1,2,3}} \rightarrow v_i$ is scheduled. 

Due to the restriction of the scheduling set, the distributed algorithm leads to a smaller equal-reciprocal stability region (say $\mathbf{\Lambda}_d$) than the centralized algorithm. Next, we will design a dynamic algorithm that achieves the  region $\mathbf{\Lambda}_d$.

The distributed algorithm is similar to Alg. \ref{alg:centralized} with some modifications as follows. 
First, we re-define the variable as
\begin{equation*}
\mu_l(t)=
\left\{
\begin{array}{ll}
0.5 & \text{if link $l=v_{a_{1,2,3}} \rightarrow v_i$ for $i=1,2,3$ is ON and scheduled;} \\
1 & \text{if other link $l$ is ON and scheduled;}\\
0 & \text{else.}
\end{array}
\right.
\end{equation*}
The notion of the re-defined $\mu_l(t)$ is that if link $v_{a_{1,2,3}} \rightarrow v_i$, for $i=1,2,3$, is scheduled at present,  the virtual link $v_i \rightarrow v_{1,2,3}^{(1,1,1)}$ must be scheduled in the next time slot; hence, the scheduled packet takes two time slots. 	Second, we re-define the set $\mathbf{L}$ as the set of all virtual links excluding $v_{i} \rightarrow v_{1,2,3}^{(1,1,1)}$ for all $i=1,2,3$. 

Then, the virtual link is scheduled as follows. If $v_{a_{1,2,3}} \rightarrow v_i$, for $i=1,2,3$, is scheduled at present, user $c_i$ needs to broadcast the received packet in the next time slot; otherwise, we select a virtual link based on Eq. (\ref{eq:back-pressure}) with the newly defined $\mu_{l}(t)$ and $\mathbf{L}$. We note that size $|\mathbf{L}|$  in the distributed algorithm is smaller, and thus the complexity of the distributed algorithm is smaller than the centralized one. 

We conclude as follows  using the \textit{frame-based} Lyapunov theorem  \cite{neely:book,neely:coding}.
\begin{theorem}
The distributed scheduling algorithm achieves the  region $\mathbf{\Lambda}_d$ and satisfies the equal-reciprocal incentive. Moreover, our distributed algorithm has two advantages over the centralized one: requiring no control message exchange with BS  and achieving  lower complexity. 
\end{theorem}

\subsection{Equal-reciprocal stability regions for the two symmetric users}
Though a lower complexity of the distributed algorithm, the distributed algorithm has  performance loss in	 the equal-reciprocal stability region. 
In this subsection, we compare the  regions of the centralized and distributed algorithms, with the focus on the local sharing part for the two symmetric users scenario, while a more general network is numerically studied in  Subsection \ref{subsection:simulation}. 


\begin{figure*}[!t]
\begin{minipage}{.3\textwidth}
\centering
\includegraphics[width=.72\textwidth]{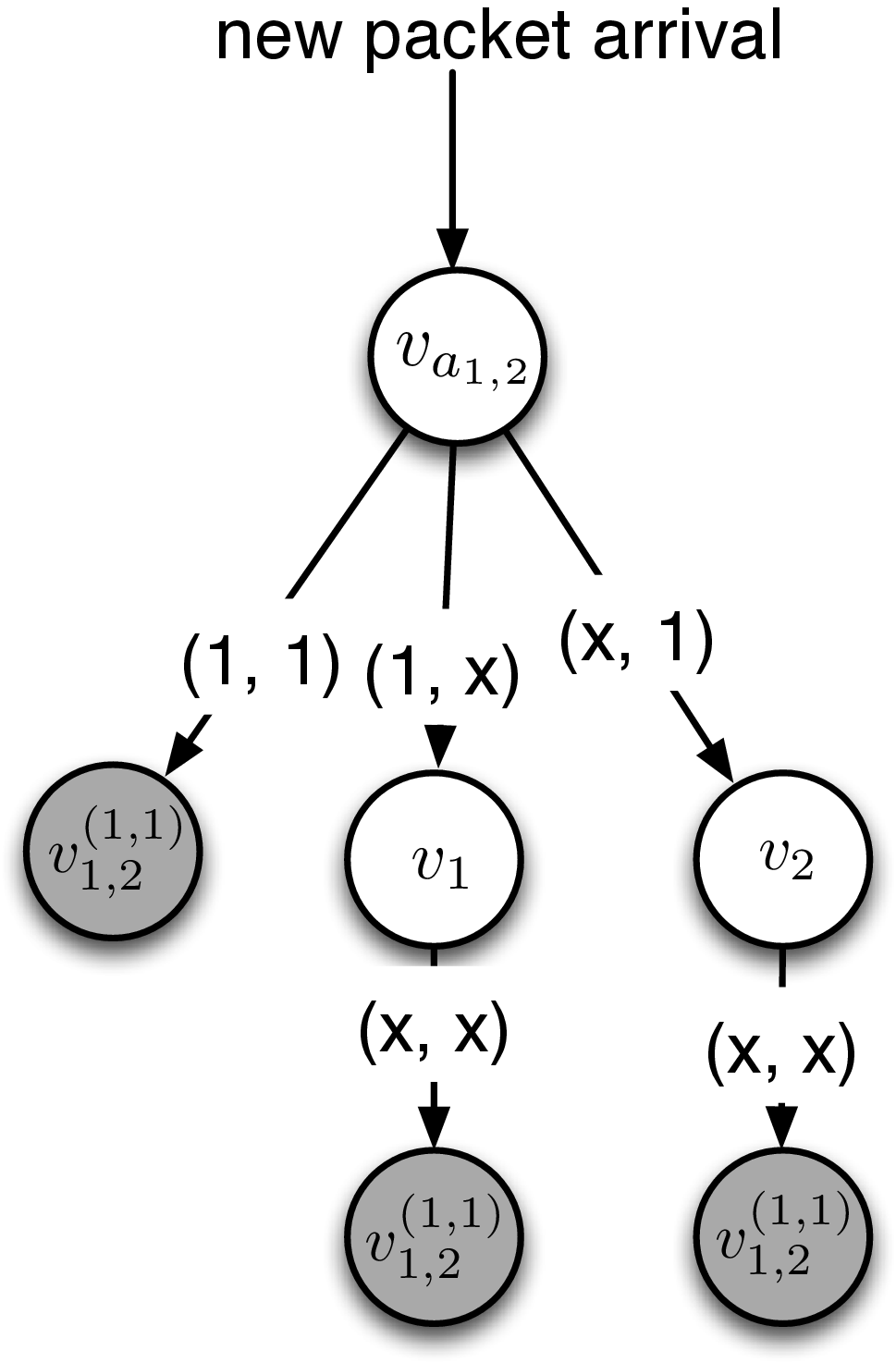}
\caption{Virtual network for the two symmetric users.}
\label{fig:two-node-capacity-1}
\end{minipage}\hfill
\begin{minipage}{.3\textwidth}
\centering
\includegraphics[width=.55\textwidth]{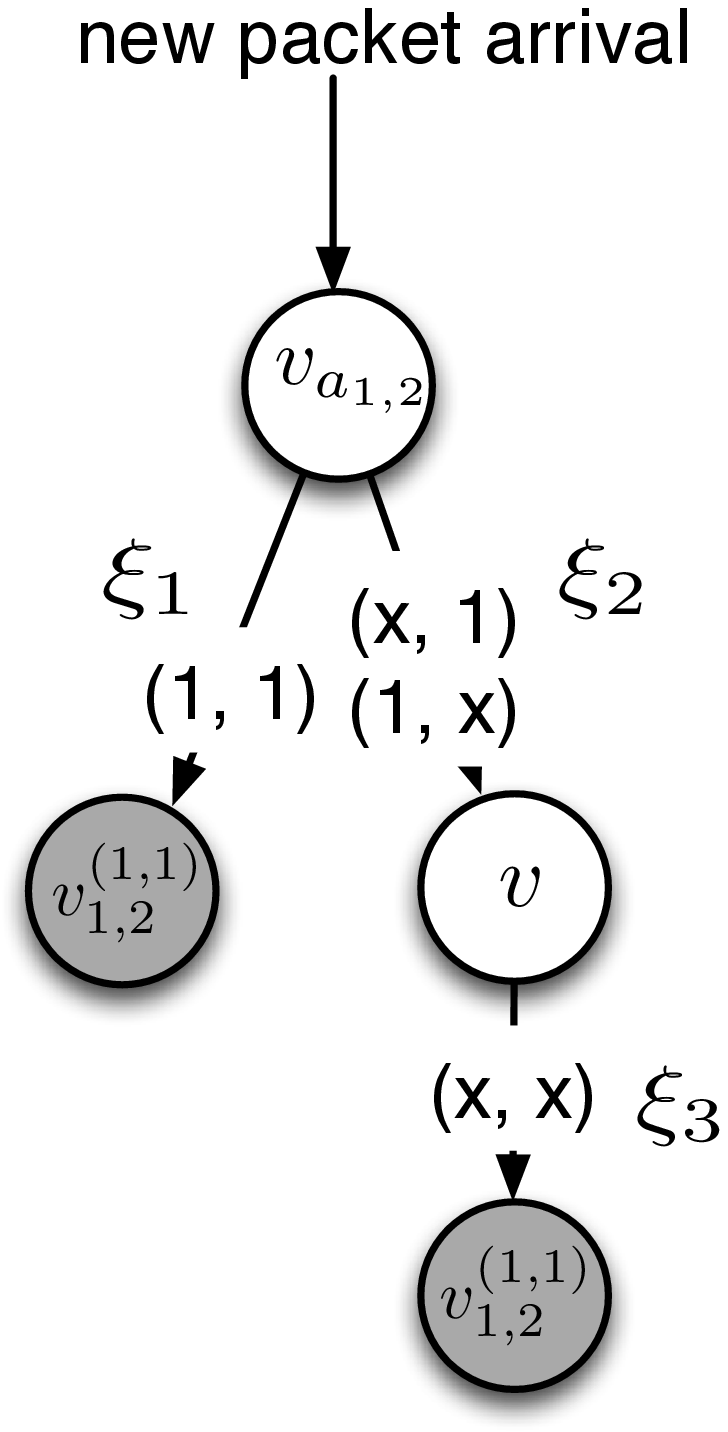}
\caption{Equivalent virtual network for the centralized algorithm.}
\label{fig:two-node-capacity-2}
\end{minipage} \hfill
\begin{minipage}{.3\textwidth}
\centering
\includegraphics[width=.6\textwidth]{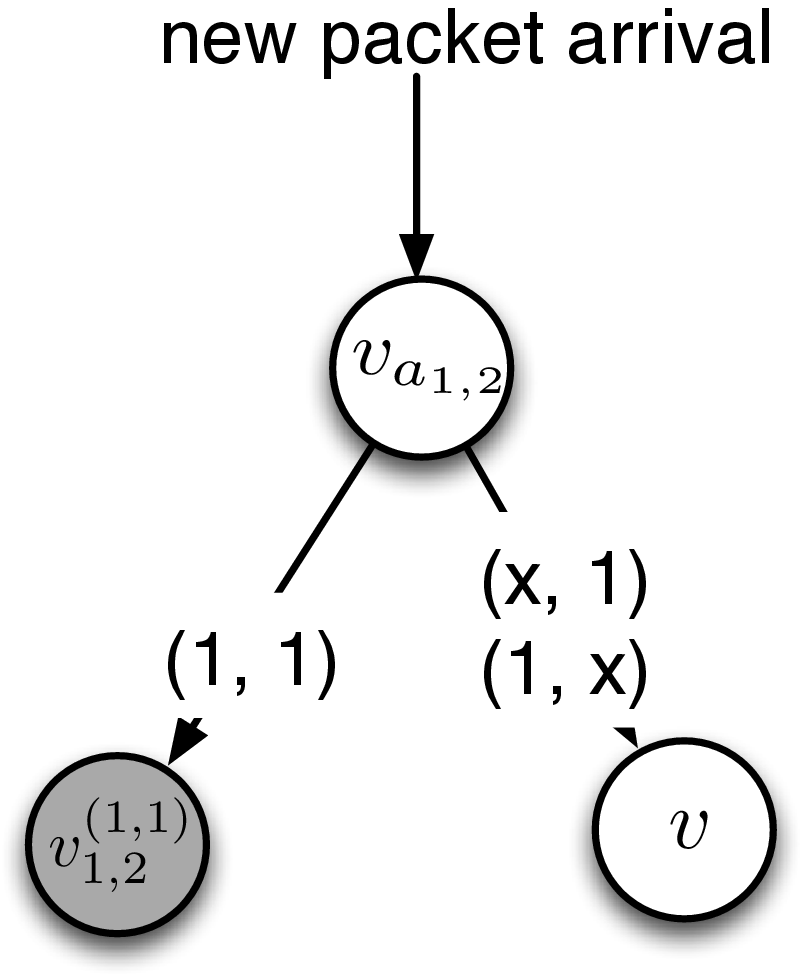}
\vspace{1cm}
\caption{Equivalent virtual network for the distributed algorithm.}
\label{fig:two-node-capacity-3}
\end{minipage}\hfill
\end{figure*}

We start with describing the equal-reciprocal stability region of the centralized algorithm. Fig. \ref{fig:two-node-capacity-1} illustrates the virtual network for the two symmetric users. Because of the symmetric users, we can combine both virtual nodes $v_{1}$ and $v_2$ as a virtual node $v$ as shown in Fig. \ref{fig:two-node-capacity-2}. 	  By $\xi_1$, $\xi_2$, and $\xi_3$ we denote the packet service rates of the links $v_{a_{1,2}} \rightarrow v_{1,2}^{(1,1)}$, $v_{a_{1,2}} \rightarrow v$, and $v \rightarrow v_{1,2}^{(1,1)}$; then, the  region can be expressed as follows.
\begin{align*}
&\lambda=\xi_1+\xi_2; \hspace{.3cm} \xi_2=\xi_3;\\
&\xi_1 \leq (1-p_e)^2; \hspace{.3cm} \xi_2\leq 2p_e(1-p_e); \hspace{.3cm} \xi_3\leq 1; \\
&\xi_1+\xi_2 \leq 1-p_e^2; \hspace{.3cm} \xi_2+\xi_3 \leq 1; \hspace{.3cm} \xi_1+\xi_3 \leq 1;\\
&\xi_1+\xi_2+\xi_3 \leq 1.
\end{align*}
These constraints come from the fact that the maximum number of packets that can be transmitted at each time is one.  By simplifying the above constraints, we derive the equal-reciprocal stability region of the centralized algorithm as follows. 
\begin{eqnarray*}
\lambda \leq 1-p_e^2.
\end{eqnarray*}

In Fig. \ref{fig:two-node-capacity-3}, we show the equivalent virtual network for the distributed algorithm, where we  combine the virtual node $v_1$, $v_2$, and $v_{1,2}^{(1,1)}$ as a virtual node $v$; as such, the time for transmitting a virtual packet over $v_{a_{1,2}} \rightarrow v$ is two as it takes into account the time for both link $v_{a_{1,2}} \rightarrow v_i$ and $v_i \rightarrow v_{1,2}^{(1,1)}$ for $i=1,2$. Therefore, we segment the timeline into frames, which size is one or two time slots if the virtual link  $v_{a_{1,2}} \rightarrow v_{1,2}^{(1,1)}$ or $v_{a_{1,2}} \rightarrow v$, respectively, is scheduled.  Similar to  \cite{neely:coding}, we then can derive the equal-reciprocal capacity region of the distributed based on algorithm  as
\begin{eqnarray*}
\lambda \cdot \left((1-p_e)^2 \cdot 1+ p_e^2 \cdot 1+2 p_e (1-p_e) \cdot 2\right) \leq 1-p_e^2,
\end{eqnarray*}
where the left term is the expected number of arrivals in a frame and the right one is the maximum number of packets that can be transmitted within a frame. 
We show both stability regions in Fig. \ref{fig:capacity} and conclude as follows. 
\begin{figure}
\centering
\includegraphics[width=.42\textwidth]{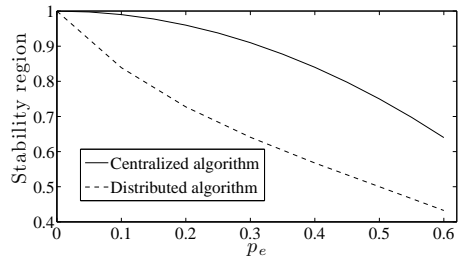}
\caption{Equal-reciprocal stability region comparison between the centralized and distributed algorithms}
\label{fig:capacity}
\end{figure}
\begin{lemma}
For the two symmetric users, the ratio of the equal-reciprocal stability region  between the centralize and the distributed algorithms is $1+2p_e-2p_e^2$. Moreover, when $p_e=1/2$, the distributed algorithm has the largest stability region loss (see Fig. \ref{fig:capacity}).  
\end{lemma}

\subsection{More numerical studies} \label{subsection:simulation}
In this subsection, we numerically study the proposed centralized and distributed algorithms for the network example in Fig. \ref{fig:dyanmics}, where we fix the channel error probabilities  of $c_1$, $c_2$,  $c_3$ to be $0.2$, $0.4$, $0.6$, respectively. We will compare the proposed algorithms with a throughput-optimal scheduling algorithm without incentivizing the social grouping, i.e., all virtual links in the local sharing part of Fig. \ref{fig:virtual} are removed.

We show the average  queue size and the average completion time of a packet in Figs. \ref{fig:queue-size} and \ref{fig:delay}. We observe that the  approximate distributed algorithm still outperforms the throughput-optimal algorithm without social grouping, and its performance is close to the centralized algorithm. 

Moreover, we  demonstrate in Fig. \ref{fig:prob} the ratio of packets that are shared among the users $c_1$, $c_2$, or $c_3$ (i.e., $\frac{\# \text{total packets shared by the users}}{\# \text{total packets at the BS}}$), which is referred to as the \textit{local content sharing probability}. We notice that the local sharing probability increases as the packet arrival rate increases.

\begin{figure*}[!ht]
\begin{minipage}{.33\textwidth}
\centering
\includegraphics[width=\textwidth]{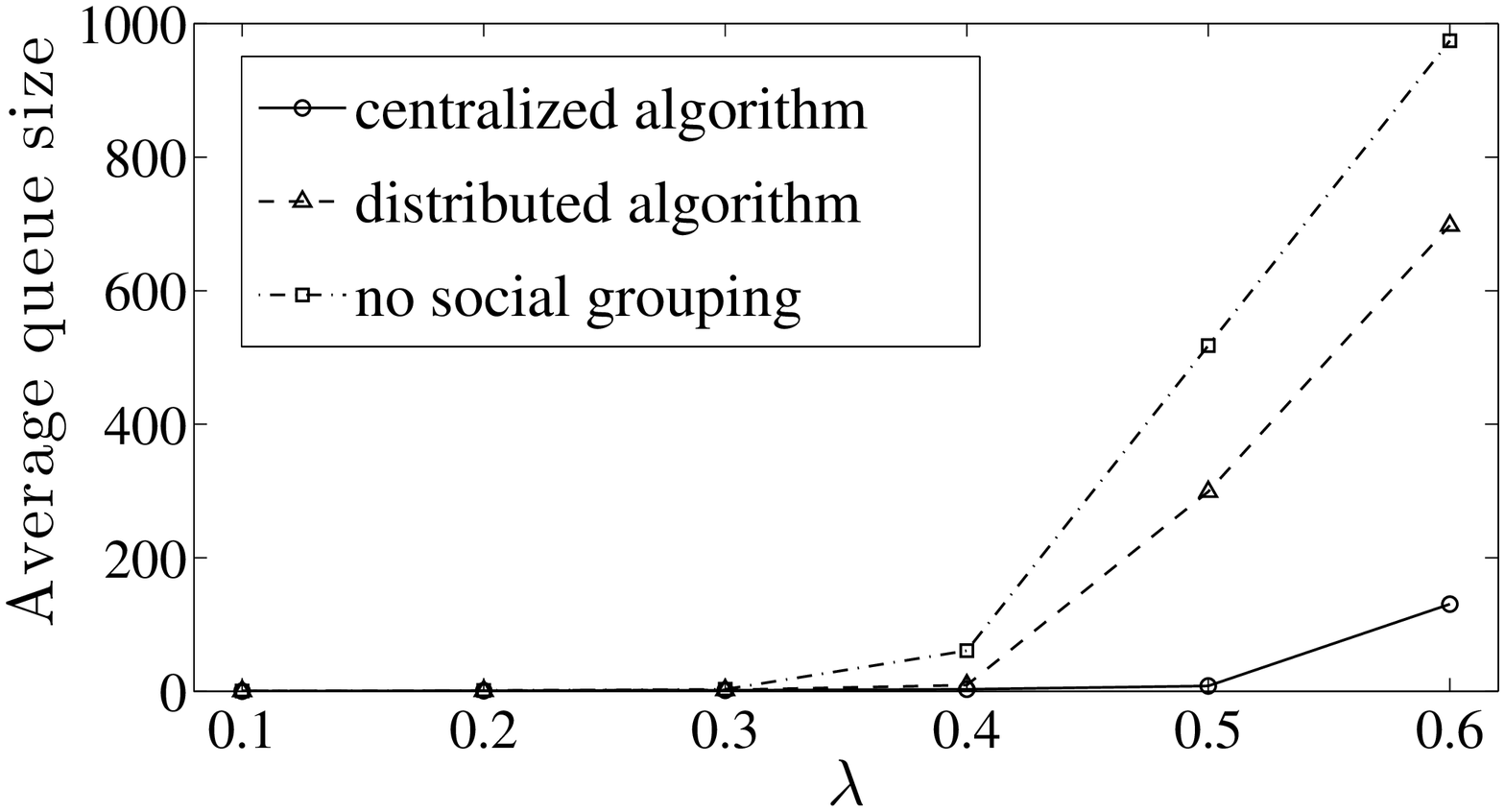}
\caption{Average  queue size versus $\lambda$.}
\label{fig:queue-size}
\end{minipage}\hfill
\begin{minipage}{.33\textwidth}
\centering
\includegraphics[width=\textwidth]{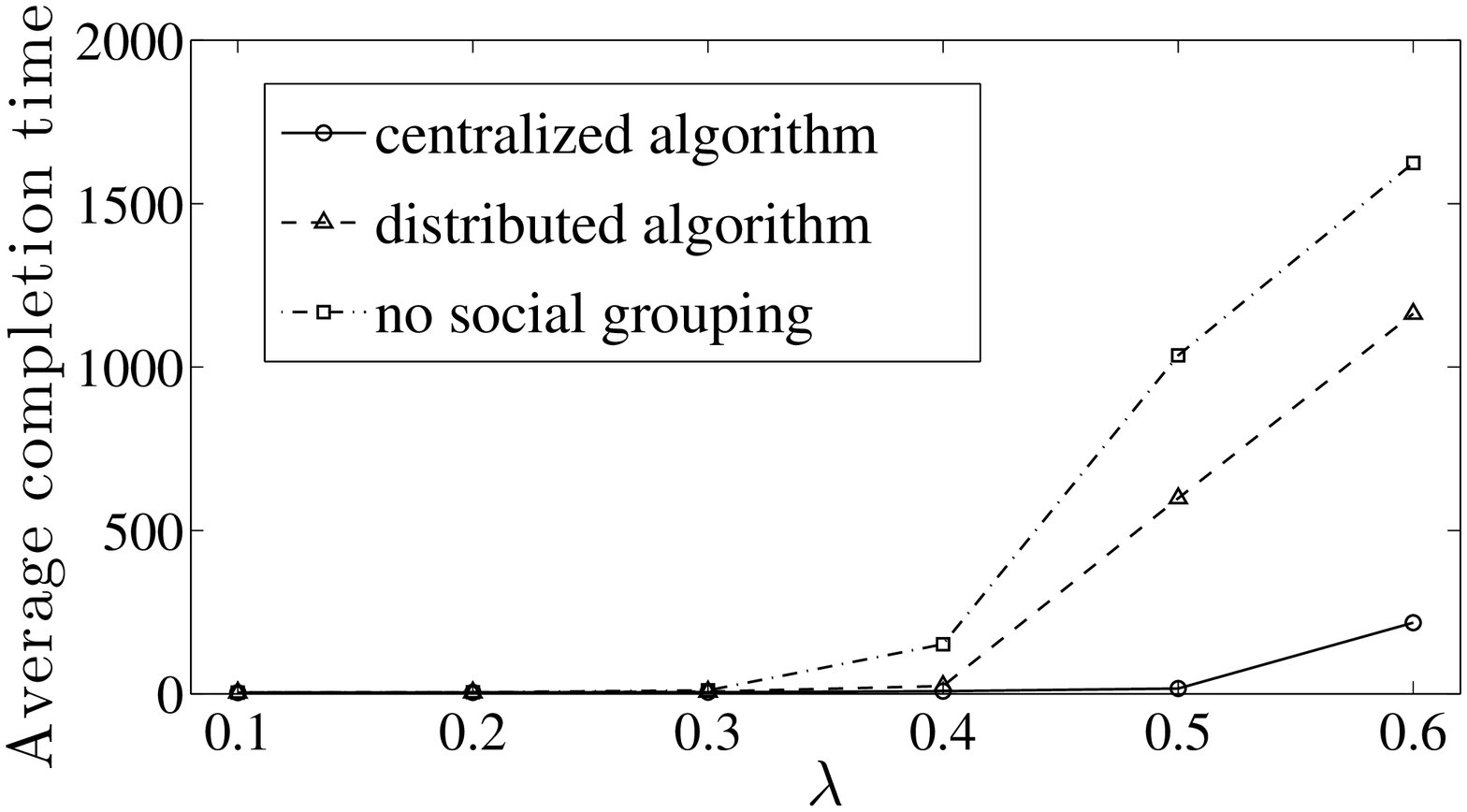}
\caption{Average completion time versus $\lambda$.}
\label{fig:delay}
\end{minipage} \hfill
\begin{minipage}{.33\textwidth}
\centering
\includegraphics[width=\textwidth]{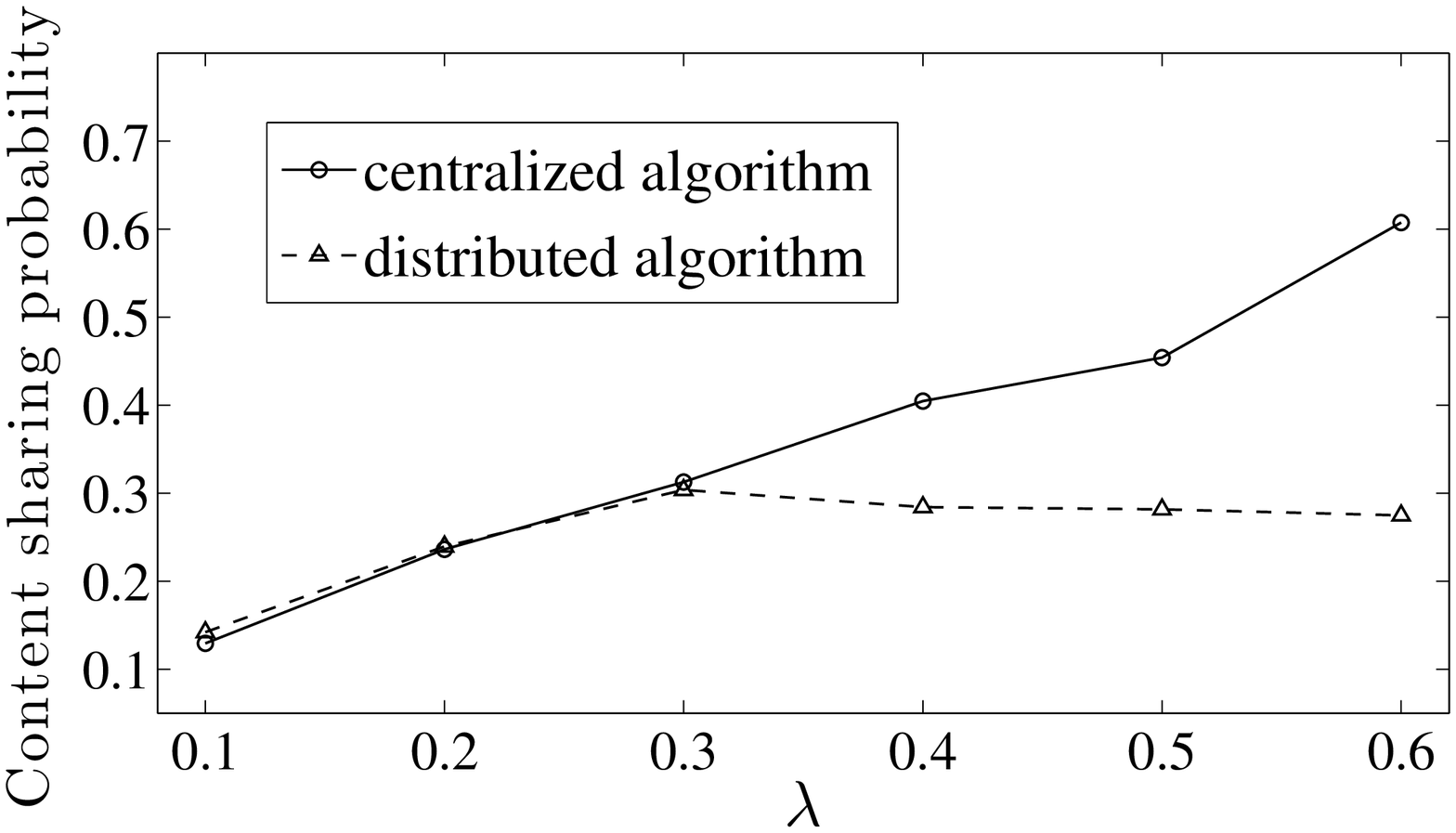}
\caption{Local content sharing probability within the social group versus $\lambda$.}
\label{fig:prob}
\end{minipage}\hfill
\end{figure*}

\section{Further discussions} \label{section:further-discussion}
We further discuss the equal-reciprocal scheme by considering correlated channel models, unreliable local communications, and a local utility of each individual user. 

\subsection{Channels with memory} \label{subsection:memory}
In this subsection, we extend our results to channel models with memory. To get a clear engineering insight, we investigate the  two symmetric users in  Section \ref{section:two-devices} again. In particular,  we consider a two-state Markov chain as the channel model, with the state $s_i(t) \in \{0, 1\}$ at time $t$ and the stationary transition probabilities $\zeta_{i,j}$ from state $i$ to $j$, i.e., $\zeta_{01}=\mathbb{P}(s_i(t+1)=1|s_i(t)=0)$, and $\zeta_{10}=\mathbb{P}(s_i(t+1)=0|s_i(t)=1)$. Then, we can get the steady-state probabilities $\pi_0=\zeta_{10}/(\zeta_{01}+\zeta_{10})$ and $\pi_1=\zeta_{01}/(\zeta_{01}+\zeta_{10})$.

Let $T_{ij}$ be the expected time from state $i$ to $j$, which can be calculated by the following equation.
\begin{align*}
T_{01}=\zeta_{00}(1+T_{01})+\zeta_{01} \cdot 1=\frac{1}{\zeta_{01}}. 
\end{align*}

We hence can calculate the completion time $T_{=}$ at the steady state: 
\begin{align*}
T_{=}&=\pi_0^2(1+T_{=})+\pi_1^2\cdot 1+2\pi_0 \pi_1 (1+\frac{1}{\zeta_{01}})=\frac{1+\frac{2 \pi_0 (1-\pi_0)}{\zeta_{01}}}{1-\pi_0^2}.
\end{align*}
Similar to Subsection \ref{subsection:incentive}, we  get $T_{\cup}^*=(-2 \pi_0^2+2 \pi_0 +1)/(1-\pi_0^2)$. To conclude, we find the improvement ratio is
\begin{align*}
R_{=/\cup}=\frac{\frac{2 \pi_0 (1-\pi_0)}{\zeta_{01}}+1}{-2 \pi_0^2+2 \pi_0 +1}.
\end{align*}

\begin{remark}
The improvement ratio is 1 when $\zeta_{01}=1$; the ratio is 1.33 when $\zeta_{01}=\zeta_{10}=0.5$; the ratio is 3 when $\zeta_{01} \rightarrow 0$, which  is consistent with Lemma \ref{lemma:improve-from-social}. In other words, our results and engineering intuition can be extended to the Markov-modulated channels. 
\end{remark}

\begin{remark}
Regarding the dynamic-arrival case,  the proposed on-line Alg. \ref{alg:centralized} can be applied for finite-state Markov chains, which can be proven using the frame-based Lyapunov technique \cite{neely:book} where the frame sizes are the transition times between the states.  
\end{remark}

\subsection{Unreliable local communications} \label{subsection:unreliable-local}
We reconsider the simple two symmetric users, but the channel between the two users can be OFF with a probability $\gamma$. To that end, we assume that the BS has to decide if a user will share a packet when receiving it. In other words, the users will  keep a copy of each packet only if being informed of sharing. We  get $T_{\cup}$:
\begin{eqnarray*}
T_{\cup}&=&p_e^2(1+T_{\cup})+(1-p_e)^2\cdot 1+2(1-p_e)p_e \left(p_{1 \rightarrow 2}\cdot (1+\frac{1}{1-\gamma})+(1-p_{1\rightarrow 2})(1+\frac{1}{1-p_e})\right).
\end{eqnarray*} 
Then, if $\gamma \leq p_e$,  the optimal sharing probability $p_{1 \rightarrow 2}^*=1$; moreover, 
\begin{align*}
T_{\cup}^*=\frac{1+\frac{2p_e(1-p_e)}{1-\gamma}}{1-p_e^2},
\end{align*}
and the improvement ratio is
\begin{align} 
R_{=/\cup}=\frac{2p_e+1}{\frac{2p_e(1-p_e)}{1-\gamma}+1}.  \label{eq:unreliable-local}
\end{align}

\begin{remark}
If $\gamma \leq p_e$, the ratio is $1$ when $p_e=0$,  and $3$ when $p_e = 1$, which is consistent with Lemma \ref{lemma:improve-from-social} again;  moreover, the ratio is $\frac{2}{0.5/(1-\gamma)+1}$ when $p_e=0.5$, where there is a performance loss due to $\gamma$. Otherwise, if $\gamma \geq p_e$, the optimal sharing probability $p_{1 \rightarrow 2}^*=0$, i.e., there is be no cooperation.  As usually local short-distance communication quality is better than long-distance communications with remote BSs, to motivate local users to construct a social group would be promising, though there will be a little bit performance loss due to the imperfect D2D (see Eq. (\ref{eq:unreliable-local})).
\end{remark}
\begin{remark}
Alg. \ref{alg:centralized} can be extended for  unreliable D2D networks, while  the  local sharing part would be similar to the  re-transmission part but the corresponding link error probabilities are associated with the error probabilities between the users. 
\end{remark}

%

\subsection{Individual performance under the optimal equal-reciprocal scheme} \label{subsection:individual}
Subject to the equal-reciprocal constraint, users are motivated to form a social group because of the fairness guarantee, while the optimal equal-reciprocal scheme minimizes the completion time.  We finally discuss whether and how the optimal equal-reciprocal scheme further benefits each individual user. 

We define a \textit{utility} $u_{i}$ of user $c_i$ as the difference  between the amount of packets downloaded from social groups and those uploaded from user $c_i$:
\begin{align*}
u_i= \text{(\#downloaded by $c_i$) - (\#uploaded by $c_i$)}. 
\end{align*}
The motivation of the utility is that long-distance B2D communications are more expensive than local D2D communications (e.g., see [2]), and therefore  a user would like to save cost by downloading content from a social group while only upload a smaller amount of content in return.

First, we notice that the tit-for-tat incentive results in the zero utility  due to  unicast communications. We show in the next theorem that the optimal equal-reciprocal scheme maximizes the local utility of each user subject to the equal reciprocal constraint.

\begin{theorem} \label{theorem:individual}
The optimal equal-reciprocal schemes (in Sections \ref{section:two-devices} and \ref{section:asymmetric}) (locally) maximize the utility $u_i$ of each individual $c_i$, for all $i$, subject to the equal-reciprocal constraint. 
\end{theorem}
\begin{proof}
Let $\alpha_{i \rightarrow j}$  be the expected amount of content user $c_i$ shares with $c_j$. Due to the equal-reciprocal constraint,  we  need $\alpha_{j \rightarrow i}=\alpha_{i \rightarrow j}$ for all $i,j$. Then, the utility $u_i$ depends on $\alpha_{i \rightarrow j}$ for all $j$. 
%

For a user $c_i$, we assume  the optimal equal-reciprocal scheme results in $\alpha^*_{i \rightarrow j}$ for all $j$. We consider two cases as follows:
\begin{itemize}
	\item $\alpha_{i \rightarrow j} < \alpha^*_{i \rightarrow j}$ for some $j$:  Because of  the broadcast medium and the equal-reciprocal constraint, each user can get at least one packet in return when broadcasting a packet. Hence, the utility is an increasing function in $\alpha_{i,j}$, i.e., $u_i(\alpha_{i \rightarrow j}) < u_i(\alpha^*_{i \rightarrow j})$.
	\item   $\alpha_{i \rightarrow j} > \alpha^*_{i \rightarrow j}$ for some $j$ : Suppose $u_i(\alpha_{i \rightarrow j}) > u_i(\alpha^*_{i \rightarrow j})$. Then, more content can be shared among users subject to the equal-reciprocal constraint, and the completion time can be further reduced.  We then get a contraction to the optimality of $\alpha^*_{i \rightarrow j}$. Hence,  $u_i(\alpha_{i \rightarrow j}) < u_i(\alpha^*_{i \rightarrow j})$. 
\end{itemize} 
By fully considering the two cases, we conclude that the optimal equal-reciprocal scheme also (locally) maximizes the local utility of each individual. 
\end{proof}

\begin{remark}
To motivate social grouping with the optimal equal-reciprocal incentive benefits not only the BS (i.e., reduce the  completion time), but also users (i.e., save cost). The optimal incentive scheme is a \textit{win-win} strategy. In particular, the optimal equal-reciprocal scheme provide two incentives for users: (1) fairness from the equal-reciprocal constraint, and (2) local utility maximization. 
\end{remark}

\begin{example}
We  now further analyze the large symmetric  network scenario in Subsection \ref{section:large-group}. Let $D$ be the expected amount of content downloaded by a user from the social group, then  
\begin{eqnarray*}
D=p_e(1-p_e^{N-1}),
\end{eqnarray*}
where the first term indicates a packet loss from the BS, while the second term is the probability that at least  one other user has got the packet. Moreover,  let $U$ be the expected amount of content uploaded by a user in the social group, then based on the sharing probability of Eq. (\ref{eq:large-network}), we get 
\begin{eqnarray*}
U=(1-p_e) \cdot \sum_{i=0}^{N-1}{N-1 \choose i}(1-p_e)^i p_e^{N-1-i} \cdot \frac{1}{i+1},
\end{eqnarray*}
where the first term implies that the user has successfully received the packet, and the variable $i$ is the number of the other users that have got the packet yet;  precisely, $1/(i+1)$ is the probability that the user is selected to share the content. 

In Fig. \ref{fig:individual}, we show the ratio of $D/U$, which indicates that a user can be rewarded with more packets from the social group than the contribution as long as the number of users is not trivial. With the aid of the incentivized social group, a user can save cost by downloading content from the  social group  while only upload a smaller amount of content in return.  Finally, we remark that using the tit-for-tat incentive, the ratio is one.
\hfill$\blacktriangleleft$

\begin{figure}
\centering
\includegraphics[width=.5\textwidth]{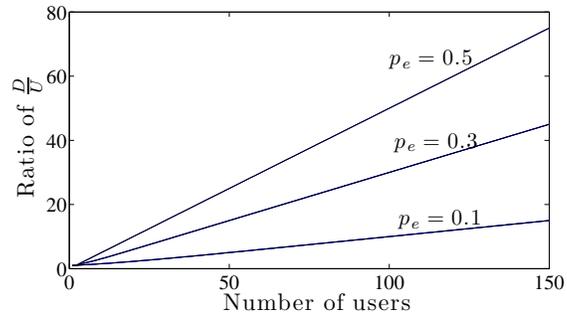}
\caption{Ratio of $D/U$ for an individual user.}
\label{fig:individual}
\end{figure} 
\end{example}

%
%
%

\section{Conclusion} \label{section:conclusion}
In this paper, we consider a practical MicroCast network scenario and propose the  equal-reciprocal incentive scheme to motivate social grouping. We theoretically investigate the optimal equal-reciprocal scheme, and  show that the optimal equal-reciprocal mechanism is a win-win policy that improves the performance of both BSs and local users. Finally, we propose on-line scheduling algorithms that dynamically select a user to share content. We conclude that the network performance can be upgraded using the proposed  social grouping and on-line algorithms. 

{\small
\bibliographystyle{IEEEtran}
\bibliography{IEEEabrv,ref}

\begin{thebibliography}{10}
\providecommand{\url}[1]{#1}
\csname url@samestyle\endcsname
\providecommand{\newblock}{\relax}
\providecommand{\bibinfo}[2]{#2}
\providecommand{\BIBentrySTDinterwordspacing}{\spaceskip=0pt\relax}
\providecommand{\BIBentryALTinterwordstretchfactor}{4}
\providecommand{\BIBentryALTinterwordspacing}{\spaceskip=\fontdimen2\font plus
\BIBentryALTinterwordstretchfactor\fontdimen3\font minus
  \fontdimen4\font\relax}
\providecommand{\BIBforeignlanguage}[2]{{%
\expandafter\ifx\csname l@#1\endcsname\relax
\typeout{** WARNING: IEEEtran.bst: No hyphenation pattern has been}%
\typeout{** loaded for the language `#1'. Using the pattern for}%
\typeout{** the default language instead.}%
\else
\language=\csname l@#1\endcsname
\fi
#2}}
\providecommand{\BIBdecl}{\relax}
\BIBdecl

\bibitem{cisco-2015}
Cisco, ``{Cisco Visual Networking Index: Global Mobile Data Traffic Forecast
  Update, 2014-2019},'' 2015.

\bibitem{google}
\BIBentryALTinterwordspacing
Google, ``{Gen V Research, Men 18–34: The On-Demand Video Consumer},'' 2012.
  [Online]. Available:
  \url{http://www.youtube.com/yt/advertise/medias/pdfs/research-gen-v-men-2.pdf}
\BIBentrySTDinterwordspacing

\bibitem{hulya-1}
L.~Keller, A.~Le, B.~Cici, H.~Seferoglu, C.~Fragouli, and A.~Markopoulou,
  ``{MicroCast: Cooperative Video Streaming on Smartphones},'' \emph{Prof. of
  ACM MobiSys}, pp. 57--70, 2012.

\bibitem{hulya-2}
A.~Le, L.~Keller, H.~Seferoglu, B.~Cici, C.~Fragouli, and A.~Markopoulou,
  ``{MicroCast: Cooperative Video Streaming Using Cellular and Local
  Connections},'' \emph{{IEEE/ACM} Trans. Netw.}, 2015.

\bibitem{hulya-3}
H.~Seferoglu, L.~Keller, B.~Cici, A.~Le, and A.~Markopoulou, ``{Cooperative
  Video Streaming on Smartphones},'' \emph{Proc. of IEEE Allerton}, vol.~59,
  no.~11, pp. 220--227, 2011.

\bibitem{social:Li}
Y.~Li, T.~Wu, P.~Hui, D.~Jin, and S.~Chen, ``{Social-Aware D2D Communications:
  Qualitative Insights and Quantitative Analysis},'' \emph{{IEEE} Commun.
  Mag.}, vol.~52, no.~6, pp. 150--158, 2014.

\bibitem{d2d-patent}
E.~Lee, S.~Kang, and J.~Chung, ``{D2D Communication Method According to D2D
  Service Type as Well as D2D Application Type, and Apparatus for Same},''
  2013, {US Patent App. 14/377,751}.

\bibitem{femto-caching}
N.~Golrezaei, K.~Shanmugam, A.~G. Dimakis, A.~F. Molisch, and G.~Caire,
  ``{FemtoCaching: Wireless Video Content Delivery through Distributed Caching
  Helpers},'' \emph{Proc. of IEEE INFOCOM}, 2012.

\bibitem{d2d-caching}
M.~Ji, G.~Caire, and A.~F. Molisch, ``{Wireless Device-to-Device Caching
  Networks: Basic Principles and System Performance},'' \emph{{IEEE} J. Sel.
  Areas Commun.}, vol.~34, no.~1, pp. 176--189, 2016.

\bibitem{maodv}
E.~Rayer and C.~Perkins, ``{Multicast Operation of the Ad Hac On-Demand
  Distance Vector Muting Protocol},'' \emph{Proc. of ACM MOBICOM}, 1999.

\bibitem{survey:liu}
J.~Liu, N.~Kato, J.~Ma, and N.~Kadowaki, ``{Device-to-Device Communication in
  LTE-Advanced Networks: A Survey},'' \emph{{IEEE} Commun. Surveys Tuts.},
  vol.~17, no.~4, pp. 1923--1940, 2015.

\bibitem{survey:Al-Kanj}
L.~Al-Kanj, Z.~Dawy, and E.~Yaacoub, ``{Energy-Aware Cooperative Content
  Distribution over Wireless Networks: Design Alternatives and Implementation
  Aspects},'' \emph{{IEEE} Commun. Surveys Tuts.}, vol.~15, no.~4, pp.
  1736--1760, 2013.

\bibitem{abedini:streamming}
N.~Abedini, S.~Sampath, R.~Bhattacharyya, S.~Paul, and S.~Shakkottai,
  ``{Realtime Streaming with Guaranteed QoS over Wireless D2D Networks},''
  \emph{Proc. of ACM MOBIHOC}, pp. 197--206, 2013.

\bibitem{cooperation:Seppala}
J.~Seppala, T.~Koskela, T.~Chen, and S.~Hakola, ``{Network Controlled
  Device-to-Device (D2D) and Cluster Multicast Concept for LTE and LTE-A
  Networks},'' \emph{Proc. of IEEE WCNC}, pp. 986--991, 2011.

\bibitem{cooperation:Andreev}
S.~Andreev, O.~Galinina, A.~Pyattaev, K.~Johnsson, and Y.~Koucheryavy,
  ``{Analyzing Assisted Offloading of Cellular User Sessions onto D2D Links in
  Unlicensed Bands},'' \emph{{IEEE} J. Sel. Areas Commun.}, vol.~33, no.~1, pp.
  67--80, 2015.

\bibitem{cooperation:lin}
X.~Lin, R.~Ratasuk, A.~Ghosh, and J.~G. Andrews, ``{Modeling, Analysis, and
  Optimization of Multicast Device-to-Device Transmissions},'' \emph{{IEEE}
  Trans. Wireless Commun.}, vol.~13, no.~8, pp. 4346--4359, 2014.

\bibitem{cooperation:Chen}
X.~Chen, X.~Gong, L.~Yang, and J.~Zhang, ``{A Social Group Utility Maximization
  Framework with Applications in Database Assisted Spectrum Access},''
  \emph{Proc. of IEEE INFOCOM}, pp. 1959--1967, 2014.

\bibitem{cooperation-coding:Liu-1}
X.~Liu, G.~Cheung, and C.-N. Chuah, ``{Structured Network Coding and
  Cooperative Wireless Ad-Hoc Peer-to-Peer Repair for WWAN Video Broadcast},''
  \emph{{IEEE} Trans. Multimedia}, vol.~11, no.~4, pp. 730--741, 2009.

\bibitem{cooperation-coding:Liu-2}
------, ``{Deterministic Structured Network Coding for WWAN Video Broadcast
  with Cooperative Peer-to-Peer Repair},'' \emph{Proc. of IEEE ICIP}, pp.
  4473--4476, 2010.

\bibitem{cooperation-coding:alex-1}
A.~Sprintson, P.~Sadeghi, G.~Booker, and S.~E. Rouayheb, ``{Deterministic
  Algorithm for Coded Cooperative Data Exchange},'' \emph{Proc. of QShine}, pp.
  282--289, 2010.

\bibitem{cooperation-coding:alex-2}
------, ``{Randomized Algorithm and Performance Bounds for Coded Cooperative
  Data Exchange},'' \emph{Proc. of IEEE ISIT}, pp. 1888--1892.

\bibitem{cooperation-coding:yupin}
I.-H. Hou, Y.-P. Hsu, and A.~Sprintson, ``{Truthful and Non-Monetary Mechanism
  for Direct Data Exchange},'' \emph{Proc. of Allerton}, pp. 406--412, 2013.

\bibitem{cooperation-incentive:Li}
J.~Li, R.~Bhattacharyya, S.~Paul, S.~Shakkottai, and V.~Subramanian,
  ``{Incentivizing Sharing in Realtime D2D Streaming Networks: A Mean Field
  Game Perspective},'' \emph{Proc. of IEEE INFOCOM}, 2015.

\bibitem{cooperation-incentive:Cao}
Y.~Cao, X.~Chen, T.~Jiang, and J.~Zhang, ``{SoCast: Social Ties Based
  Cooperative Video Multicast},'' \emph{Proc. of IEEE INFOCOM}, pp. 415--423,
  2014.

\bibitem{p2p:neely}
M.~J. Neely, ``{Optimal Peer-to-Peer Scheduling for Mobile Wireless Networks
  with Redundantly Distributed Data},'' \emph{{IEEE} Trans. Mobile Comput.},
  vol.~9, pp. 2086--2099, 2014.

\bibitem{p2p:Kamvar}
S.~D. Kamvar, M.~T. Schlosser, and H.~Garcia-Molina, ``{The Eigentrust
  Algorithm for Reputation Management in P2P Networks},'' \emph{Proc. of ACM
  WWW}, pp. 640--651, 2003.

\bibitem{p2p:fabil}
F.~Pianese, D.~Perino, J.~Keller, and E.~W. Biersack, ``{PULSE: An Adaptive,
  Incentive-Based, Unstructured P2P Live Streaming System},'' \emph{{IEEE}
  Trans. Multimedia}, vol.~9, no.~8, pp. 1645--1660, 2007.

\bibitem{p2p:zhou}
R.~Zhou and K.~Hwang, ``{Powertrust: A Robust and Scalable Reputation System
  for Trusted Peer-to-Peer Computing},'' \emph{{IEEE} Trans. Parallel Distrib.
  Syst.}, vol.~18, no.~4, pp. 460--473, 2007.

\bibitem{net-coding}
A.~Eryilmaz, A.~Ozdaglar, M.~Medard, and E.~Ahmed, ``{On the Delay and
  Throughput Gains of Coding in Unreliable Networks},'' \emph{{IEEE} Trans.
  Inf. Theory}, vol.~54, pp. 5511--5524, 2008.

\bibitem{channel-model:i-hong}
I.-H. Hou and P.~R. Kumar, ``{Broadcasting Delay-Constrained Traffic over
  Unreliable Wireless Links with Network Coding},'' \emph{Proc. of ACM
  MOBIHOC}, 2011.

\bibitem{channel-model:yang}
Y.~Yang and N.~Shroff, ``{Throughput of Rateless Codes over Broadcast Erasure
  Channels},'' in \emph{Proc. of ACM MOBIHOC}, 2012, pp. 125--134.

\bibitem{estimation}
S.~Kay, \emph{{Fundamentals of Statistical Signal Processing: Estimation
  Theory}}.\hskip 1em plus 0.5em minus 0.4em\relax Prentice Hall, 1993.

\bibitem{neely:book}
L.~Georgiadis, M.~J. Neely, and L.~Tassiulas, \emph{{Resource Allocation and
  Cross-Layer Control in Wireless Networks}}.\hskip 1em plus 0.5em minus
  0.4em\relax Now Publishers Inc, 2006.

\bibitem{neely:coding}
M.~J. Neely, A.~S. Tehrani, and Z.~Zhang, ``{Dynamic Index Coding for Wireless
  Broadcast Networks},'' \emph{{IEEE} Trans. Inf. Theory}, vol.~59, no.~11, pp.
  7525--7540, 2013.

\end{thebibliography}
}

\appendices
\section{B2D broadcast by exploiting the common interest} \label{appendix:b2d-identify}
Though traditionally the users communicate with the BS via their cellular links \textit{independently},  one of the users in a group can initiate a streaming session; then, the BS will identify the common interest and deliver the common interest over the broadcast B2D; moreover, a new user can join existing streams. These functions have been implemented in the \textit{MicroDownload} system \cite{hulya-1,hulya-2,hulya-3}, which is a component of the MicroCast system. We now analyze the benefit from the common interest in the time-varying ON/OFF network.

Let $T_{\neq}$ be the  completion time if the BS uses the traditional unicast to deliver $q$ to $c_1, c_2$ (i.e., without identifying the common interest).  We notice that $T_{\neq} < 2/(1-p_e)$ as the BS knows the channel state vector in advance and hence can schedule a packet to transmit if the associated channel is ON. Then, we express $T_{\neq}$ as follows.
\begin{eqnarray}
T_{\neq}&=&p^2_e(1+T_{\neq})+(1-p_e^2)(1+\frac{1}{1-p_e}) \label{eq:Td}\\
&=&\frac{p_e+2}{1-p_e^2}, \nonumber
\end{eqnarray}
where the first term in RHS of Eq. (\ref{eq:Td})  indicates that both channels are OFF and the first B2D transmission takes one time slot, while the second term means that one of the packets $q_1$ and $q_2$ has been successfully delivered and the other has not yet; hence, it takes the expected time of $1/(1-p_e)$ to deliver the other packet.  

Hence, the improvement from identifying the common interest (see Eq. (\ref{eq:Tc})) is 
\begin{eqnarray*}
R_{\neq/=}:=\frac{T_{\neq}}{T_{=}}=\frac{p_e+2}{2p_e+1}. 
\end{eqnarray*} 

\section{To incentivize a user or not} \label{section:to-incentivize}
As a user with a poor B2D channel reduces sharing content, \textit{when there is a new coming user, is it the best choice to motivate the user to join in the social group, or is it better to serve the user only by the BS?} We remind  that  users of a social group would  exchange a key before  any content transfer; as such, users outside the social group cannot decode received packets from the social group.

%
%
%
%
%
%
%

First, we analyze three asymmetric users $c_1$, $c_2$, and $c_3$ like Subsection \ref{subsection:large-asymmetric}, while $(c_1, c_2)$ is an existing social group and $c_3$ is a new user, where $p_{e,1}$, $p_{e,2}$, $p_{e,3}$ are the channel error probabilities, respectively. Via the simplified case, we get more insight; a general case will be discussed later. 

We consider two possible groupings $G_1=((c_1, c_2),(c_3))$ and $G_2=(c_1, c_2, c_3)$. If $c_3$ is not motivated to participate in the existing group, the BS needs to transmit the packet to $c_3$. Let $T^*_{G_1}$ and $T^*_{G_2}$ be the minimum completion times to deliver the common interest to $c_1$, $c_2$, $c_3$ with the grouping $G_1$ and  $G_2$, respectively.  We note that $T^*_{G_2}$ is calculated in Subsection \ref{subsection:large-asymmetric}.

First, we consider that $p_{e,3} \geq \max(p_{e,1}, p_{e,2})$. There will be a trade-off according to  two cases: 
\begin{itemize}
	\item \textit{Increase of user diversity}: As  the benefit to exploit user diversity increases (i.e., there users $c_1$, $c_2$, and $c_3$ in $G_2$ can help  each other but only $c_1$ and $c_2$ help  each other in $G_1$), the optimal completion time $T^*_{G_2}$ might be shorter   than $T^*_{G_1}$; 
	\item \textit{Decrease of sharing probability}: Similar to Section \ref{section:asymmetric}, adding a worse user $c_3$ to the existing group  reduces the sharing  probabilities between $c_1$ and $c_2$ because of the equal-reciprocal constraint. 
\end{itemize}
In other words, adding $c_3$ to the existing social group  increases the user diversity, while reducing the sharing probability. The trade-off between the user diversity and the sharing probability leads to a question:  \textit{should $c_3$  be grouped}? 

Second, we notice that  if $p_{e,3} < \max(p_{e,1}, p_{e,2})$,  then $G_2$ is better because adding $c_3$ to the existing social group  increases both the user diversity and the sharing probabilities. Therefore, we focus on $p_{e,1} \leq p_{e,2} \leq p_{e,3}$. 

%
%
%
%
%
%
%
%
%

Similar to Subsection \ref{subsection:large-asymmetric} and Theorem  \ref{theorem:linear-program}, we can calculate the ratio  $T^*_{G_1}/T^*_{G_2}$, which is larger than one as shown in Fig. \ref{fig:grouping}. From the numerical results, we find incentivizing the new coming $c_3$ minimizes the completion time. 

\begin{figure}
\centering
\includegraphics[width=.5\textwidth]{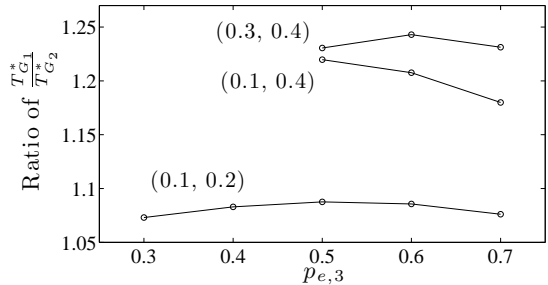}
\caption{The ratio of $T^*_{G_1}/T^*_{G_2}$ for various pairs of $(p_{e,1}, p_{e,2})$.}
\label{fig:grouping}
\end{figure}

While a more rigorous proof will be in the next subsection,  we are discussing an intuition here. Because $p_{e,3} \geq \max(p_{e,1}, p_{e,2})$, the time when $c_3$ receives the content dominates the completion time even though the sharing probabilities between $c_1$ and $c_2$ decrease when including $c_3$. Moreover, since $c_1$ and $c_2$ share content with $c_3$ in $G_2$, the user $c_3$ in $G_2$ will get the packet earlier. In other words, $G_2$  minimizes the completion time.

\subsection{Large asymmetric networks} \label{subsection:to-motivate-large-asymmertic}
We now consider a general case where there exists a social group  $(c_1, \cdots, c_N)$ and a new user $c_{N+1}$. Let $p_{e,i}$, $1 \leq i \leq N+1$, be their channel error probabilities, respectively. Without loss of generality, we assume that $p_{e,1} \leq \cdots \leq p_{e, N+1}$. Similarly, we denote two groupings $G_1=((c_1, \cdots, c_N),c_{N+1})$ and $G_2=(c_1, \cdots, c_{N+1})$. We start with a theorem regarding $p_{e, N+1} \geq 0.5$ as below.
\begin{theorem}
If $p_{e, N+1} \geq 0.5$,  to group $c_{N+1}$ minimizes the completion time, i.e., to group all is the best. 
\end{theorem}
\begin{proof}
We consider an auxiliary group $\hat{G}=(\hat{c}_1, \cdots, \hat{c}_{N+1})$, with  the error probabilities of each user $\hat{c}_i$, $i=1, \cdots, N+1$, being $p_{e, N+1}$.  Let $T^*_{\hat{G}}$ be the corresponding minimum completion time, and $T_{c_{N+1}}$ be the expected time when $c_{N+1}$ successfully receives the packet from the BS when $G_1$ is applied. Then, we have
\begin{align}
T^*_{G_1} \geq & T_{c_{N+1}} \label{eq:group-all-1}\\
\geq & T^*_{\hat{G}} \label{eq:group-all-2}\\
\geq & T^*_{G_2}, \label{eq:group-all-3} 
\end{align} 
where Eq. (\ref{eq:group-all-1}) is because $c_{N+1}$ should be satisfied; Eq. (\ref{eq:group-all-2}) is according to Lemma \ref{theorem:large-symmetric}; Eq. (\ref{eq:group-all-3}) is because all nodes in $\hat{G}$ have the  worse channels than those in $G_{2}$. 
\end{proof}

We remind that the reason why $G_1$ is better than $G_2$ is that including a user $c_{N+1}$ with a worse channel reduces the sharing probabilities between $c_1, \cdots, c_N$. Moreover, in group $G_2$ the sharing probabilities increase when $p_{e, N+1}$ decreases, while in group $G_1$ the sharing probabilities is the same for all $p_{e, N+1}$ value. Hence, when $p_{e, N+1} < 0.5$, the grouping $G_2$ is better than $G_1$ as the sharing probabilities is larger than the case of $p_{e, N+1}=0.5$, in which $G_2$ is better.     We therefore can conclude the optimality of grouping all as follows. 

\begin{theorem}
To incentivize a new coming users minimizes the completion time, no matter what the distribution of B2D channel errors is. 
\end{theorem}



\end{document}